\documentclass[a4paper,11pt]{amsproc}
\usepackage{amssymb}
\usepackage{amscd}
\usepackage[dvips]{graphicx}
\usepackage[all,cmtip]{xy}
\usepackage[inline]{enumitem}
\usepackage{xcolor}
\usepackage{cite}
\usepackage[hidelinks]{hyperref}

\theoremstyle{plain}
 \newtheorem{thm}{Theorem}[section]
 \newtheorem{prop}{Proposition}[section]
 \newtheorem{lem}{Lemma}[section]
 
\theoremstyle{definition}
 \newtheorem{exm}{Example}[section]
 \newtheorem{rem}{Remark}[section]
\newtheorem{dfn}{Definition}[section]

\numberwithin{equation}{section}

\newcommand{\N}{\mathcal{ N}}

\newcommand{\R}{\mathbb{R}}

\DeclareMathOperator{\rank}{\mathrm{rank}}

\allowdisplaybreaks

\title[The principle of relativity in classical nonholonomic mechanics]
{The principle of relativity in classical nonholonomic mechanics and $G$-Chaplygin systems with affine time-dependent constraints}
\subjclass[2020]{37J06, 37J60, 70F25, 70G45,  70H40}
\author[Jovanovi\'c]{%\bfseries
Bo\v zidar Jovanovi\'c}

\address{
Mathematical Institute SANU \\
Serbian Academy of Sciences and Arts \\
Kneza Mihaila 36, 11000 Belgrade\\
Serbia}
\email{bozaj@mi.sanu.ac.rs}

\begin{document}

\begin{abstract}
We study the principle of relativity in classical time-dependent nonholonomic Lagrangian mechanics from the perspective of moving global reference frames. We also introduce time-dependent $G$-Chaplygin systems with affine constraints, which are natural objects for the invariant formulation of nonholonomic systems with symmetries.
As far as the author is aware, the Hamiltonization problem for time-dependent constraints has not yet been studied. As a first step in this direction, we consider a rolling without sliding of a balanced disc of radius $r$ over a vertical circle of  a given  variable radius $R(t)$. We modify the Chaplygin multiplier method and prove that the reduced system becomes Hamiltonian with respect to the new time.
\end{abstract}

\maketitle

\tableofcontents

\section{Introduction}

We study the principle of relativity in  classical time-dependent nonholonomic Lagrangian mechanics. We follow \cite{Jo2} where we have treated holonomic systems by using analogies to fixed and moving reference frames in rigid body dynamics.
It turns out that for a space-time formulation of nonholonomic mechanics it is natural to pass from a class of systems with homogeneous nonholonomic constraints defined by distributions of the tangent bundle of the configuration space to a class of nonholonomic systems with nonhomogeneous constraints. That is why we consider a constrained Lagrangian system $(Q,L,\mathcal A)$, where $Q$ is an $n$--dimensional configuration space, $L\colon TQ\times\R\to\R$ is a time-dependent Lagrangian and a motion is subject to time-dependent non-homogeneous constraints: the velocity $\dot\gamma(t)$ of an admissible motion $\gamma(t)$ belongs to a time-dependent affine distribution $\mathcal A_t$ of the tangent bundle $TQ$.

A motion of the system is derived from the \emph{d'Alambert principle} or the \emph{d'Alambert-Lagrange principle}: the trajectories of a mechanical system are obtained from the condition that the variational derivative of the Lagrangian vanishes along virtual displacements \cite{AKN, Bloch}.
Recall that Galileo Galilei's (1564--1642) \emph{principle of relativity} is one of the most important steps towards understanding nature (in modern terms  of inertial frames, see  Poincar\'e \cite{P1} and Arnold \cite{Ar}). It is interesting that the
  d'Alambert principle
  fulfils the following general variant of the Galilean principle of relativity, which does not include a notion of inertial frames (see e.g. \cite{Jo2}):

\begin{itemize}

\item  \emph{All the laws of nature at all moments of time are the same in all reference frames.}

\end{itemize}

With appropriate geometric structures on space-time manifolds and a notion of reference frames, both general relativity (with local  reference frames) and classical mechanics (with global reference frames, see Section \ref{s4}) agree with the above principle. The invariant formulation of classical nonholonomic Lagrangian mechanics on a space-time manifold is well known (see e.g. \cite{MP} and references therein), but is not so emphasised. On the other hand, the above principle, incorporated in Einstein's general equivalence principle, was one of the basic motivations for its foundation, and it is widely used in general relativity.

We intend to clarify time-dependent affine nonholonomic systems using as classical an approach as possible, building upon the principle of relativity, but from the point of view of global reference frames.
All objects can be defined without local coordinates, thereby going beyond the standard covariance of the d'Alembert equations in different coordinate systems or reference frames. We describe the transformations of Lagrangians, external forces, and the distribution of affine constraints with respect to arbitrary time-dependent diffeomorphisms of configuration spaces (Section 3). Furthermore, in the second part of the paper, we consider time-dependent $G$-Chaplygin systems by using the classical rheonomic Voronec equations (Sections \ref{s5} and \ref{s6}).

Although the geometry of time-dependent affine nonholonomic systems, including the Voronec equations, has already been extensively studied using jet bundles (e.g., see \cite{Leon1997, MP}) and Lie algebroids (e.g., see \cite{Iglesias2007, Balseiro2010}), to the author's knowledge, the model presented in Section \ref{s6} with the explicit Hamiltonization of the reduced rheonomic nonholonomic system is the first example of its kind in the existing literature on time-dependent nonholonomic systems.

\subsection{Outline and results of the paper}

In Section \ref{s2} we recall the d'Alambert principle for a nonholonomic Lagrangian system $(Q,L,\mathcal A)$, which also includes the field of non-potential forces $\mathbf F$. Motivated by the notion of fixed and moving reference frames in rigid body dynamics \cite{Ar, AKN, Jo2}, we consider arbitrary time-dependent transformations between the configuration space $Q$ (the fixed reference frame) and the manifold $M$ diffeomorphic to $Q$ (the moving reference frame) and consider trajectories of a nonholonomic Lagrangian system in both reference frames (Theorem \ref{posledica2}). A notion of moving energy \cite{BMB, FS} naturally appears in the
invariant formulation of nonholonomic mechanics (Section \ref{s3}). All considerations are valid without the assumption that the Lagrangian is regular and are derived without the use of Lagrange multipliers.

In Section \ref{s4} we apply the construction of moving reference frames for the definition of nonholonomic Lagrangian mechanics in a space-time, $(n+1)$--dimensional manifold $\mathcal Q$, which is fibred over $\R$ with fibers diffeomorphic to $Q$. The invariant formulation of time-dependent classical Lagrangian mechanics is well studied (see e.g. \cite{MP} and references therein). Here,  following \cite{Jo2},
we have tried to present it with minimal technical requirements.

In Section \ref{s5} we consider nonholonomic systems $(Q,L,\mathcal A)$ on  fibred spaces and in Section \ref{s6} we use them to describe the reduction of time-dependent $G$--Chaplygin systems associated to time-dependent principal bundles. Recall that the usual $G$--Chaplygin systems have a natural geometric framework as connections on principal bundles (see \cite{Koi1992, EK2019}). On the other hand, nonholonomic systems with symmetries, in particular Chaplygin systems, are incorporated into the geometric framework of Ehresmann connections on fibred spaces (see \cite{BKMM}). In this paper, following \cite{DGJ2}, we combine the approach of \cite{BKMM} with the Voronec nonholonomic equations \cite{Vor1902} and derive an invariant form of the Voronec equations for time-dependent Ehresmann connections (Theorem \ref{invVoronec}, Proposition \ref{movingVoronec}).

The invariant form of the equations allows us to introduce time-dependent $G$-Chaplygin systems with affine constraints and to perform a $G$--Chaplygin reduction
(Section \ref{s6}, Theorem \ref{redukcijaSistema}).
A naturally related problem is the Hamiltonization of $G$--Chaplygin systems. For natural mechanical systems, the reduced system has more additional terms compared to the usual case of homogeneous time-independent constraints (see Section \ref{s6}). As far as the author is aware, the Hamiltonization problem for time-dependent constraints has not yet been studied. As a first step in this direction, we consider a rolling without sliding of a balanced disc of radius $r$ over a vertical circle of
a given variable radius $R(t)$. We modify the Chaplygin multiplier method and prove that the reduced system becomes the usual Lagrangian system with respect to the new time (Theorem \ref{hamiltonizacija*}).

\section{D'Alambert principle} \label{s2}

\subsection{Nonholonomic systems with affine constraints}

We consider a nonholonomic Lagrangian system $(Q,L,\mathcal A)$, where $Q$ is an $n$--dimensional configuration space, $L(q,\dot q,t)$ is a time--dependent Lagrangian, $L: TQ \times \R\to \R$, and $\mathcal A$ are nonholonomic constraints
\[
\mathcal A=\{(\mathcal A_t,t)\,\vert\, t\in\R\}\subset TQ \times \R,
\]
where $\mathcal A_t$ is a time-dependent affine distribution of rank $m$ of the tangent bundle $TQ$. A curve $\gamma\colon\R\to Q$ is \emph{admissible} (or allowed by constraints) if the velocity $\dot\gamma(t)$ belongs to $\mathcal A_t$, $t\in\R$.

The affine distribution can be written in the following form
\[
\mathcal A_t=\mathcal D_t+\chi_t,
\]
where $\mathcal D_t$ is a time-dependent distribution of rank $m$ and $\chi_t$ is a time-dependent vector field on $Q$ defined modulo $\mathcal D_t$. The vectors in $\mathcal A_t$ are called \emph{admissible velocities}, while the vectors in $\mathcal D_t$ are called \emph{virtual displacements}.

Let $r=n-m$. In local coordinates $q=(q^1,\dots,q^n)$ on $Q$, the constraints are given by equations
\begin{equation}\label{constraints}
\dot q\in \mathcal A_t\vert_q \Longleftrightarrow \sum_{i=1}^n a_{\nu i}(q,t)\dot q^i+ a_\nu(q,t)=0, \qquad  \nu=1,\dots,r,
\end{equation}
while virtual displacements satisfy
\[
\eta\in \mathcal D_t\vert_q \Longleftrightarrow \sum_{i=1}^n a_{\nu i}(q,t)\eta^i=0, \qquad  \nu=1,\dots,r,
\]
where $\rank(a_{\nu i})_{1\le \nu\le r,1\le i \le n}=r$.

We assume that the constraints are nonholonomic. This means that there are no functions $f_\nu(q^1,\dots,q^n,t)$, $\nu=1,\dots,r$, such that
the affine distribution $\mathcal A_t$ is locally defined by \eqref{constraints}, where
\[
a_{\nu}=\frac{\partial f_\alpha}{\partial t}, \qquad a_{\nu i}=\frac{\partial f_\nu}{\partial q^i}, \qquad \nu=1,\dots,r, \qquad i=1,\dots,n.
\]

That is why it is convenient to
consider the rank $(m+1)$ distribution $\underline{\mathcal A}\subset T(Q\times\R)$ of the \emph{extended configuration space} $Q\times\R$ (see Fig. \ref{distribucija}) defined by
\begin{equation}\label{podvucenoA}
\underline{\mathcal A}_{(q,t)}=\big\{\lambda\eta+\mu\big(\chi_t+\frac{\partial}{\partial t}\big)\, \big\vert\,
\eta\in \mathcal D_t\vert_q, \lambda,\mu\in\R \big\}\subset T_{(q,t)}(Q\times\R),
\end{equation}
or in local coordinates
\begin{equation*}\label{constraints*}
\underline{\eta}=\sum_{i=1}^n \eta^i\frac{\partial}{\partial q^i}+\eta^{n+1}\frac{\partial}{\partial t} \in \underline{\mathcal A}_{(q,t)}
\Longleftrightarrow \sum_{i=1}^n a_{\nu i}(q,t)\eta^i+ a_\nu(q,t)\eta^{n+1}=0, \qquad  \nu=1,\dots,r.
\end{equation*}
According to the Frobenius theorem (e.g., see \cite{LM}), if $\underline{\mathcal A}$
is nonintegrable, the constraints are nonholonomic.

\begin{figure}[ht]
{\centering
{\includegraphics[width=7cm]{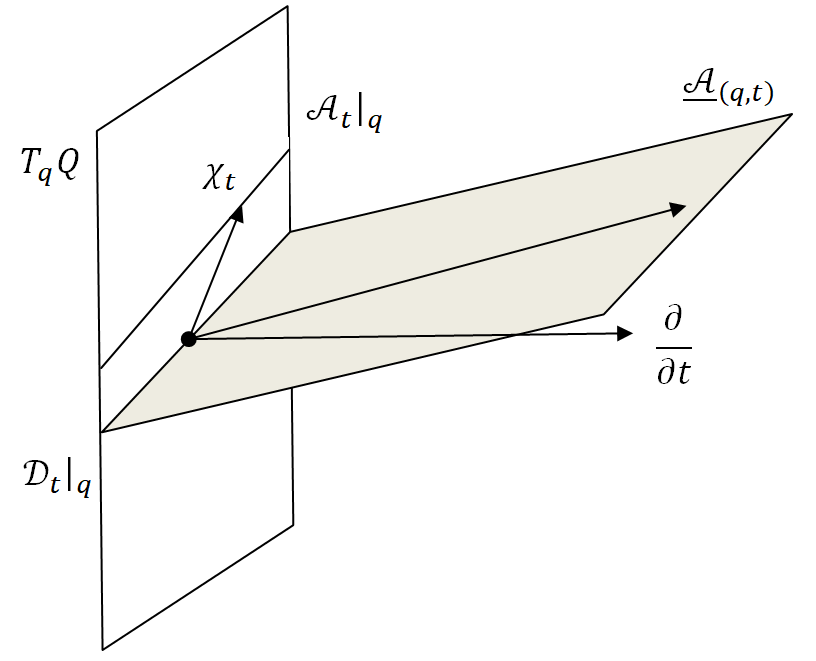}}
\caption{The tangent space $T_{(q,t)}(Q\times\R)$ of the extended configuration space and its subspaces
$\mathcal D_t\vert_q, \mathcal A_t\vert_q\subset T_q Q$ and $\underline{\mathcal A}_{(q,t)}$.}\label{distribucija}}
\end{figure}

\begin{exm}\label{primer1}
Consider $Q=\R^2\{x,y\}$ and the nonhomogeneous constraint
\begin{equation}\label{primer}
\dot x - a(x,y,t) \dot y-b(x,y,t)=0.
\end{equation}
Then
\begin{align*}
&{\mathcal D}_t\vert_{(x,y)}=
\big\{ \xi=\xi^{1}\frac{\partial}{\partial x}+\xi^{2}\frac{\partial}{\partial y}\in T_{(x,y)}\R^2 \,\big\vert\,
\xi^{1} =a(x,y,t) \xi^{2}\big\},\quad \chi_t=b(x,y,t)\frac{\partial}{\partial x},\\
&{\mathcal A}_t\vert_{(x,y)}=
\big\{ \xi=\xi^{1}\frac{\partial}{\partial x}+\xi^{2}\frac{\partial}{\partial y}\in T_{(x,y)}\R^2 \,\big\vert\,
\xi^{1} =a(x,y,t) \xi^{2}+b(x,y,t)\big\},\\
& \underline{\mathcal A}_{(x,y,t)}=\big\{{\underline{\xi}=} \xi^{1}\frac{\partial}{\partial x}+\xi^{2}\frac{\partial}{\partial y}+
\xi^{3}\frac{\partial}{\partial t}\in T_{(x,y,t)}\R^3 \,\big\vert\,
\xi^{1} =a(x,y,t) \xi^{2}+b(x,y,t)\xi^{3}\big\}.
\end{align*}

Since $\underline{\mathcal A}$ is a kernel of 1-form $\alpha=dx-ady-bdt$ on the extended configuration space $\R^3\{(x,y,t)\}$, the constraint \eqref{primer} is nonholonomic if $\alpha$ is a contact:
\begin{align*}
\alpha\wedge d\alpha=&\big( dx-ady-bdt\big )\wedge \big(\frac{\partial a}{\partial t} dy \wedge dt-\frac{\partial a}{\partial x} dx \wedge dy - \frac{\partial b}{\partial y} dy \wedge dt- \frac{\partial b}{\partial x} dx \wedge dt\big)\\
=&\big(\frac{\partial a}{\partial t}-\frac{\partial b}{\partial y}+b\frac{\partial a}{\partial x}-a\frac{\partial b}{\partial x}\big)dx\wedge dy \wedge dt \ne 0.
\end{align*}
In particular, even for homogeneous constraints $b\equiv 0$, if ${\partial a}/{\partial t}\ne 0$ the constraints are nonholonomic, although the distribution $\mathcal A_t=\mathcal D_t$ is integrable for any fixed $t$.
\end{exm}

\subsection{Dynamics}

In classical mechanics, the dynamics in the case of \emph{ideal nonholonomic constraints} is defined by the \emph{d'Alembert principle} (e.g, see \cite{AKN, Bloch}): an admissible curve $\gamma$ is a motion of the constrained Lagrangian system $(Q,L,\mathcal A)$ if the variational derivative $\delta L(\eta)\vert_\gamma$ vanishes for all virtual displacements $\eta$ along $\gamma$:
\begin{equation*}\label{dalamber}
\delta L(\eta)\vert_\gamma=
\sum_{i=1}^n \big(\frac{\partial L}{\partial q^i}-\frac{d}{dt}\frac{\partial L}{\partial \dot q^i}\big)\eta^i\vert_{\gamma(t)}=0, \qquad \eta\vert_{\gamma(t)}\in \mathcal D_t\vert_{\gamma(t)}.
\end{equation*}
For integrable constraints, the principle is equivalent to another fundamental principle -- \emph{Hamiltonian principle of least action}, see e.g. \cite{AKN, Bloch, Jo2}.

If we also have a field of non-potential forces
$\mathbf F(q,\dot q,t)$,
\[
\mathbf F(q,\dot q,t)=\sum_{i=1}^n F_i(q,\dot q,t) dq^i,
\]
the equations of motion are given by
\begin{equation}\label{dalamberF}
\delta L(\eta)\vert_\gamma+\mathbf F(\eta)\vert_\gamma=
\sum_{i=1}^n \big(\frac{\partial L}{\partial q^i}-\frac{d}{dt}\frac{\partial L}{\partial \dot q^i}+F_i\big)\eta^i\vert_{\gamma(t)}=0, \qquad \eta\vert_{\gamma(t)}\in \mathcal D_t\vert_{\gamma(t)}.
\end{equation}

Consider the associated time-dependent fiber derivative
$\mathbb FL: TQ\to T^*Q$ defined by
\begin{equation}\label{legendre}
\mathbb FL(q,\xi,t)(\eta)=
\frac{d}{ds}\vert_{s=0}L(q,\xi+s\eta,t), \qquad \xi,\eta\in T_q Q.
\end{equation}

If \eqref{legendre} is a diffeomorphism between $TQ$ and $T^*Q$ for all $t\in\R$, the corresponding Lagrangian $L$ is called \emph{regular}.
Then the initial value problem $q(t_0)=q_0$, $\dot{q}(t_0)=\dot{q}_0$ of the system has the unique solution. Locally, we can write the d'Alambert principle \eqref{dalamberF} in the form of Euler-Lagrange equations with multipliers
\begin{equation}\label{euler-lagrange}
\frac{d}{dt}\frac{\partial L}{\partial \dot q^i}=\frac{\partial L}{\partial q^i}+F_i+\sum_\nu \lambda_\nu a_{\nu i}, \qquad i=1,\dots, n,
\end{equation}
where the Lagrange multipliers $\lambda_\nu=\lambda_\nu(1,\dot q,t)$ for regular Lagrangians are uniquely determined from the condition that a motion $\dot q$ satisfies the constraints \eqref{constraints}.

Let $(q^1,\dots,q^n, p_1,\dots,p_n)$ be the
canonical coordinates of the cotangent bundle $T^*Q$. Locally, the fiber derivative \eqref{legendre} is given by
\begin{equation*}\label{legendre*}
\quad p_i=\frac{\partial L}{\partial \dot q^i}, \qquad i=1,\dots, n.
\end{equation*}

For regular Lagrangians it is convenient to consider $(n+m+1)$--dimensional submanifold
\[
\mathcal M=\{(\mathcal M_t=\mathbb FL(\mathcal A_t),t)=(q^1,\dots,q^n,p_1,\dots,p_n,t)\,\vert\,p_i=\frac{\partial L}{\partial \dot q^i},\, \dot q\in\mathcal A_t, \, t\in\R\}
\]
of the extended phase space $T^*Q\times\R$.

The system \eqref{euler-lagrange} transforms into a first order system on $\mathcal M$ of the form:
\begin{align*}
&\dot q^i=\frac{\partial H}{\partial p_i}, \qquad\qquad\qquad i=1,\dots,n,\\
&\dot p_i=-\frac{\partial H}{\partial q^i}+\big(F_i+\sum_\nu \lambda_\nu(q,\dot q,t)a_{\nu i}(q,t)\big)\vert_{\dot q=\mathbb FL^{-1}(p)},
\end{align*}
where the \emph{Hamiltonian function} $H(q,p,t)$ is the \emph{Legendre
transformation} of $L$
\begin{equation*} \label{haml}
 H(q,p,t)=E(q,\dot q,t)\vert_{\dot q=\mathbb FL^{-1}(q,p,t)}=\mathbb FL(q,\dot q,t)\cdot \dot q-L(q,\dot
q,t)\vert_{\dot q=\mathbb FL^{-1}(q,p,t)}.
\end{equation*}

For non-regular Lagrangian systems see e.g. \cite{LMM, LOS}.

\section{Nonholonomic systems in the moving frames} \label{s3}

\subsection{The invariance of the d'Alambert principle}
In analogy to rigid body dynamics, where we define the fixed and the moving frame by time-dependent isometries of Euclidean space, we consider \emph{the moving reference frame} in a general Lagrangian system $(Q,L)$ as a time-dependent diffeomorphism
\[
g_t\colon M\to Q, \qquad q=g_t(x).
\]
Here $M=Q$, but we use different symbols to underline the domain and codomain of the mapping: the variable $q$ is in the fixed reference frame, while the variable $x$ is in the moving reference frame. Furthermore, in analogy to the angular velocity in the fixed and in the moving frame, we define the time-dependent vector fields $\omega\in \mathfrak X(Q)$ and $\Omega\in\mathfrak X(M)$ by the identities (see \cite{Jo2})
\begin{equation}\label{omegaOmega}
\begin{aligned}
&\omega_t(q)=\frac{d}{ds}\vert_{s=0}\big(g_{t+s}\circ (g_t)^{-1}(q)), \\
&\Omega_t(x)=\frac{d}{ds}\vert_{s=0}\big((g_{t})^{-1}\circ (g_{t+s})(x)).
\end{aligned}
\end{equation}

Note that $g_t$ as a curve in the Lie group $\text{\sc Diff}(Q)$ and the vector fields $\omega_t$ and $\Omega_t$ are elements of its Lie algebra $\mathfrak X(Q)=Lie(\text{\sc Diff}(Q))$ given by the right and left translation of the velocity $\dot g_t$ (see also \cite{AK}).
%\footnote{Here we are not interested in the smooth structures of these infinite-dimensional objects.}
They are related by the adjoint mapping: $\omega_t=\mathrm{Ad}_{g_t}(\Omega_t)$.

\begin{exm}
Consider a motion of a rigid body in $\R^3$ and let $g_t\colon \R^3\{\vec x\}\to \R^3\{\vec q\}$, $\vec q=g_t(\vec x)=\vec r(t)+B_t(\vec x)$  be a mapping from the frame attached to the body to the frame attached in space ($B_t\in SO(3)$, $g_t\in SE(3)$). If $\vec\Gamma(t)$ is a motion in the moving frame and $\vec\gamma(t)=g_t(\vec\Gamma(t))$ is the corresponding motion in fixed space, we have (e.g., see \cite{Ar})
\begin{align*}
\dot{\vec\gamma}=& B_t(\dot{\vec\Gamma}(t))+\dot{\vec r}(t)+\dot B_t B^{-1}_t(\vec\gamma(t)-\vec r(t))\\
=& B_t(\dot{\vec\Gamma}(t))+\dot{\vec r}(t)+\vec\omega(t)\times (\vec\gamma(t)-\vec{r}(t)),
\end{align*}
where $\vec\omega(t)$ is the angular velocity of the body in the fixed reference frame. The group of Euclidean motions $SE(3)$ is a finite-dimensional subgroup of $\text{\sc Diff}(\R^3)$. The time-dependent vector field $\omega_t\in\mathfrak X(\R^3)$ in the fixed frame associated to the curve $g_t\in \text{\sc Diff}(\R^3)$ by \eqref{omegaOmega} is given by
\[
\omega_t(\vec q)=\dot{\vec r}(t)+\vec\omega(t)\times (\vec q-\vec r(t)).
\]
%If $\vec r(t)\equiv 0$ (motion of a body around a fixed point), we simply have
%$\omega_t(\vec q)=\vec\omega(t)\times \vec q$.
%Note that in \cite{FA}  rigid body dynamics is considered from the perspective of continuum  mechanics.
\end{exm}

As in the case of the rigid body we have (see e.g. \cite{Jo2}):

\begin{prop}\label{sabiranjeBrzina}
(i) The angular velocity vector fields are related by
\begin{equation*}\label{ugaone}
dg_t(\Omega_t)\vert_x=\omega_t\vert_{q=g_t(x)}.
\end{equation*}

(ii) {\sc The classical addition of velocities.} Let $\Gamma(t)$ be a smooth curve on $M$ and $\gamma(t)=g_t(\Gamma(t))$ be the corresponding curve on $Q$. Then
\[
\dot\gamma{(t)}=dg_t(\dot\Gamma{(t)}))+\omega_t(g_t(\Gamma(t))).
\]
Conversely, for a given curve $\gamma(t)$ and the corresponding curve $\Gamma(t)=g_t^{-1}(\gamma(t))$, we have
\[
\dot\Gamma{(t)}=dg_t^{-1}(\dot\gamma{(t)}))-\Omega_t(g_t^{-1}(\gamma(t))).
\]
\end{prop}

For a given Lagrangian $L\colon TQ\times \R\to \R$ we define the associated Lagrangian $l\colon TM\times \R\to \R$ in the moving frame by
\begin{equation*}\label{noviLagranzijan}
l(x,\dot x,t):=L(q,\dot q,t)\vert_{q=g_t(x),\dot q=dg_t(\dot x)+\omega_t(g_t(x))}.
\end{equation*}

The following observation is fundamental in what follows (e.g., see \cite{Jo2}).

\begin{prop}\label{glavna}
Let $\Gamma$ and $\xi$ be a smooth curve and a time-dependent vector field in the moving frame and $\gamma(t)=g_t(\Gamma(t))$ and
$\eta_t=dg_t(\xi_t)$ be the associated curve and the time-dependent vector field in the fixed frame. Then the variational derivative of $L$ along $\gamma$ in the direction of $\xi$ coincides with the variational derivative of $l$ along $\Gamma$ in the direction of $\eta$:
\begin{equation*}\label{invDal}
\delta L(\eta)\vert_\gamma=\delta l(\xi)\vert_\Gamma.
\end{equation*}
\end{prop}

Let us define the distribution $\mathbf D_t$ of the virtual displacements and the distribution of the admissible velocities $\mathbf A_t$ in the moving frame by the identities
\begin{equation}\label{pokretneDistribucije}
\mathbf D_t\vert_x =dg^{-1}_t(\mathcal D_t\vert_q), \qquad \mathbf A_t\vert_x=dg^{-1}_t(\mathcal A_t\vert_q)-\Omega_t\vert_x, \qquad x=g^{-1}_t(q).
\end{equation}
Then
\begin{align*}
&\mathbf A_t=\mathbf D_t+X_t, \qquad X_t\vert_x=dg^{-1}_t(\chi_t\vert_q)-\Omega_t\vert_x, \qquad x=g^{-1}_t(q), \\
\text{i.e.,}\qquad&\mathcal A_t\vert_q=dg_t(\mathbf A_t\vert_x)+\omega_t\vert_q, \qquad \chi_t\vert_q=dg_t(X\vert_x)+\omega_t\vert_q, \qquad q=g_t(x).
\end{align*}

Let $\mathbf A=\{(\mathbf A_t,t)\,\vert\,t\in\R\}$ and define a field of non-potential forces in the moving frame by
\begin{equation}\label{pokretneSile}
\mathbf f(q,\dot q,t)(\xi):=\mathbf F(x,\dot x,t)(dg_t(\xi)),
\end{equation}
where
\[
\xi\in T_x M, \qquad  q=g_t(x),\qquad \dot q=dg_t(\dot x)+\omega_t(g_t(x)).
\]

\begin{thm}\label{posledica2}
A curve $\gamma(t)$ is a motion of the nonholonomic Lagrangian system $(Q,L,$ $\mathbf F,\mathcal A)$ if and only if $\Gamma(t)=g_t^{-1}(\gamma(t))$ is a motion
of the nonholonomic Lagrangian system $(M,l,\mathbf f,\mathbf A)$.
\end{thm}

 \begin{proof}
Assume that $\gamma(t)$ is an admissible curve ($\dot\gamma(t)\in\mathcal A_t\vert_{\gamma(t)}$) that satisfies the d’Alembert principle \eqref{dalamberF} for a virtual displacement $\eta_t$. From Propositions \ref{sabiranjeBrzina} and \ref{glavna}, the definitions of the affine distribution $\mathbf A_t$ (see \eqref{pokretneDistribucije}) and the non-potential forces $\mathbf f$ (see \eqref{pokretneSile}), it follows that
$\Gamma(t)=g_t^{-1}(\gamma(t))$ is an admissible curve that satisfies the d’Alembert principle
\[
\delta l(\xi)\vert_\Gamma+\mathbf f(\xi)\vert_\Gamma=0,
\]
for a virtual displacement $\xi_t\vert_{\Gamma(t)}=dg_t^{-1}(\eta_t\vert_{\gamma(t)})$. Since $\mathbf D_t\vert_{\Gamma(t)}=dg_t^{-1}(\mathcal D_t\vert_{\gamma(t)})$, the statement is proved.
\end{proof}

Analogous to \eqref{podvucenoA}, we define the distribution $\underline{\mathbf A}\subset T(M\times\R)$ of the extended moving configuration space
\begin{equation}\label{PodvucenoMasnoA}
\underline{\mathbf A}_{(x,t)}=\big\{\lambda\xi+\mu\big(X_t+\frac{\partial}{\partial t}\big)\, \big\vert\,
\xi\in \mathbf D_t\vert_x, \lambda,\mu\in\R \big\}\subset T_{(x,t)}(M\times\R).
\end{equation}

Together with $g_t$ we consider the mapping
\[
\phi\colon M\times\R \longrightarrow Q\times \R, \qquad (q,t)=\phi(x,t)=(g_t(x),t).
\]

Then %$\chi_t+\frac{\partial}{\partial t}=d\phi\vert_{(x,t)}(X_t+\frac{\partial}{\partial t})$, and therefore
\[
\underline{\mathcal A}=d\phi(\underline{\mathbf A}).
\]

We can interpret the above relations in such a way that the distributions of the virtual displacements $\mathcal D_t$ and the distribution $\underline{\mathcal A}$ of the extended configuration space are "geometric objects".   %while the affine distribution $\mathcal A_t$ is not a "geometrical object".

If the original constraints are homogeneous, they are generally inhomogeneous in the moving frame.
Conversely, the affine distribution $\mathcal A_t$ can be a distribution in the moving frame (the nonhomogeneous constraints can be homogeneous in the moving frame).
Therefore, in using the principle of relativity in nonholonomic mechanics, it is natural to pass from a class of systems with homogeneous nonholonomic constraints to a class of nonholonomic systems with nonhomogeneous constraints.

If $\mathbf D_t=\mathbf A_t$, we say that the moving frame $M$ is a \emph{distinguished reference frame}.
We have
\begin{align*}
\mathbf A_t=\mathbf D_t &\Longleftrightarrow \Omega_t=dg^{-1}_t(\chi_t)+\xi_t\\
&\Longleftrightarrow
\omega_t=\chi_t+\eta_t,
\end{align*}
for some virtual displacement vector field $\xi_t$ on $M$ ($\eta_t$ on $Q$).
Locally we can always find a time-dependent transformation $g_t$ such that $\chi_t=\omega_t$, so that locally distinguished reference frame always exist.

\begin{exm}\label{primer2}
In natural mechanical problems, nonholonomic constraints are usually defined by the non-slip condition.
For example, consider a ball that rolls over a rotating table without slipping. In the moving reference frame attached to the rotating plane, the constraints are homogeneous, while in the fixed frame they are nonhomogeneous.\end{exm}

\subsection{Natural mechanical systems and moving energies}\label{MovEn}

Recall the definition of the energy of the Lagrangian system with the Lagrangian $L$:
\begin{equation*}\label{Energy}
E_L(q,\dot q,t)=\mathbb FL(q,\dot q,t)(\dot q)-L(q,\dot
q,t)=\sum_{i}\frac{\partial L}{\partial \dot q_i}\dot q_i-L.
\end{equation*}
For regular Lagrangians $E_L=H\circ \mathbb FL$.
It is known that if the Lagrangian $L$ does not depend on time and the time-dependent constraints are homogeneous, the nonholomic system
preserves the energy (e.g., see Proposition 1.3.3 in \cite{Bloch}).

\begin{rem}
One of the first examples of nonholonomic systems with the energy integral and time-dependent homogeneous nonholonomic constraints is a time-dependent variant of the Suslov nonholonomic rigid body motion, which was  defined by Bilimovi\'c \cite{Bilimovic1913b, Bilimovic1915}, see also \cite{BT2020}.
\end{rem}

It should be noted that all of the above considerations apply without the assumption that the Lagrangian is regular, and that they were derived without the use of Lagrange multipliers. Let us now consider a \emph{natural mechanical system} with constraints $(Q,L,\mathcal A)$. The Lagrangian $L$ has the form
\[
L(q,\dot q,t)=\frac12 K_t(\dot q,\dot q)+\Delta_t(\dot q)-V(q,t),
\]
where $K_t$ is a time-dependent Riemannian metric, $\Delta_t$ a time-dependent 1-form and $V$ a potential function.
Since
\begin{equation}\label{FB}
\mathbb FL(q,\dot q,t)(\cdot)=K_t(\dot q,\cdot)+\Delta_t(\cdot),
\end{equation}
the Lagrangian is regular and the energy is the sum of the kinetic and potential energy
\[
E_L(q,\dot q,t)=\frac12 K_t(\dot q,\dot q)+V(q,t).
\]

\begin{rem}
The natural mechanical system $(Q,L,\mathcal A)$
with the Lagrangian with linear term $\Delta_t(\dot q)$ is equivalent to the natural mechanical system with additional non-potential forces
 $(Q,L_0,\mathcal A,\mathbf F)$, where
\[
L_0(q,\dot q,t)=\frac12 K_t(\dot q,\dot q)-V(q,t),
\]
and
\[
\mathbf F(q,\dot q,t)(\eta)=d \Delta\big(\eta,\dot q+\frac{\partial}{\partial t}\big)=d \Delta_t\big(\eta,\dot q\big)-\sum_i\frac{\partial\Delta_i(q,t)}{\partial t}\eta^i, \qquad \eta\in T_q Q.
\]
Here, a time-dependent one-form
$
\Delta_t=\sum_i \Delta_i(q^1,\dots,q^n,t) dq^i
$
is also regarded as a one-form $\Delta$ on the extended configuration space $Q\times\R$ as well. Then
\[
d\Delta=\sum_{i<j}\big(\frac{\partial\Delta_j(q,t)}{\partial q^i}
-\frac{\partial\Delta_i(q,t)}{\partial q^j}\big) dq^i\wedge dq^j
+\sum_{i}\frac{\partial\Delta_i(q,t)}{\partial t} dt \wedge dq^i
\]
is the differential on $Q\times\R$, while
\[
d\Delta_t=\sum_{i<j}\big(\frac{\partial\Delta_j(q,t)}{\partial q^i}
-\frac{\partial\Delta_i(q,t)}{\partial q^j}\big) dq^i\wedge dq^j
\]
is the differential on $Q$ for a fixed $t$.
Thus, $\mathbf F$ consists of the gyroscopic and dissipative term $\mathbf F=\mathbf F_{gyr}+\mathbf F_{dis}$:
\begin{align*}
&\mathbf F_{gyr}(q,\dot q,t)=\sum_{i,j}\big(\frac{\partial\Delta_j(q,t)}{\partial q^i}
-\frac{\partial\Delta_i(q,t)}{\partial q^j}\big)\dot q^j dq^i, \\
& \mathbf F_{dis}(q,\dot q,t)=
-\sum_{i}\frac{\partial\Delta_i(q,t)}{\partial t} dq^i.
\end{align*}
Usually the Lagrangian $L$ does not depend on time and only the gyroscopic (or magnetic) term is considered (see e.g. \cite{DGJ2} and references therein).
\end{rem}

In a moving frame, the Lagrangian has by definition the form
\[
l(x,\dot x,t)=\frac12 \kappa_t (\dot x,\dot x)+\delta_t(\dot x)-v(x,t),
\]
where the kinetic energy, the linear term and the potential energy are given by
\begin{align*}
&\kappa_t(\xi,\eta)\vert_x=K_t(dg_t(\xi),dg_t(\eta))\vert_{q=g_t(x)},  \\
& \delta_t(\xi)\vert_x=\Delta_t(dg_t(\xi))+K_t(\omega_t,dg_t(\xi))\vert_{q=g_t(x)}=g_t^*\Delta_t(\xi)+\kappa_t(\Omega_t,\xi)\vert_x,\\
&v(x,t)=V(q,t)-\frac12 K_t(\omega_t,\omega_t)-\Delta_t(\omega_t)\vert_{q=g_t(x)}=g_t^*V-\frac12 \kappa_t(\Omega_t,\Omega_t)-g_t^*\Delta_t(\Omega_t)\vert_x,
\end{align*}
while the corresponding energy becomes
\[
E_l(x,\dot x,t)=\frac12 \kappa_t (\dot x,\dot x)+v(x,t).
\]

In \cite{BMB, FS, FNS} the following simple but important observation is used in the study of conservation of energy in nonholonomic systems.

\begin{prop}
The energies $E_L$ and $E_l$ are related by
\[
E_L(q,\dot q,t)\vert_{q=g_t(x),\dot q=dg_t(\dot x)+\omega_t(g_t(x))}=E_l(x,\dot x,t)+\mathbb Fl(x,\dot x,t)(\Omega_t)
\]
and vice versa,
\[
E_l(x,\dot x,t)\vert_{x=g^{-1}_t(q),\dot x=dg^{-1}_t(\dot q)-\Omega_t(g^{-1}_t(x))}=E_L(q,\dot q,t)-\mathbb FL(q,\dot q,t)(\omega_t).
\]
\end{prop}

\begin{proof}
The first identity follows from
\begin{align*}
E_L(q,\dot q,t)\vert_{q=g_t(x),\dot q=dg_t(\dot x)+\omega_t(g_t(x))}&=\frac12\kappa_t(\dot x,\dot x)+\kappa_t(\Omega_t,\dot x)+\frac12 \kappa_t(\Omega_t,\Omega_t)+g_t^*V\\
&=\frac12\kappa_t(\dot x,\dot x)+v+\kappa_t(\Omega_t,\dot x)+\kappa_t(\Omega_t,\Omega_t)+g_t^*\Delta_t(\Omega_t)\\
&=\frac12\kappa_t(\dot x,\dot x)+v+\kappa_t(\Omega_t,\dot x)+\delta_t(\Omega_t)\\
&=E_l(x,\dot x,t)+\mathbb Fl(x,\dot x,t)(\Omega_t).
\end{align*}
If we interchange  the roles of the moving and fixed frame, we obtain the second identity.
\end{proof}

In the case that the nonholonomic system in the moving frame $(M,l,\mathbf A)$
has the energy $E_l$ that does not depend on time and the constraints $\mathbf A$ are homogeneous, i.e. $M$ is a distinguished reference frame, then the energy $E_l$ is conserved (in \cite{BMB, FS,FNS} the somewhat stronger assumption that the constraints are time-independent is considered). Therefore, the function
\[
J(q,\dot q,t)=E_L(q,\dot q,t)-\mathbb FL(q,\dot q,t)(\omega_t)
\]
is the integral of the system $(Q,L,\mathcal A)$ in the fixed reference frame.

Slightly more generally, if there exist a time-dependent vector field $\xi_t$ such that
\[
J(q,\dot q,t)=E_L(q,\dot q,t)-\mathbb FL(q,\dot q,t)(\xi_t)
\]
is conserved along trajectories of nonholonomic system $(Q,L,\mathcal A)$, then we say that $J$ is the \emph{moving energy} (or the \emph{Jacobi energy} integral),
see \cite{FS, BMB}.
On the cotangent bundle side, the moving energy corresponds to
\[
I(q,p,t)=H(q,p,t)-p(\xi_t).
\]

The moving energy can be considered within the framework of the Noether symmetries of time-dependent nonholonomic systems \cite{Jo, NS}.
Another classical approach to a notion of energy for time-dependent Lagrangian systems can be found in \cite{Mus1, Mus2}.

\section{Space-time formulation of nonholonomic mechanics}\label{s4}

\subsection{Space-time and reference frames}

Recall that depending on the geometric structure, Galilean or pseudo--Euclidian, of the 4-dimensional affine space-time  $\mathcal W^4$ (\emph{world}), we obtain two different mechanics: classical and special relativity, which share the same Galilean principle of relativity (see e.g. \cite{Ar, Jo}).
Following the Galilean structure in the affine world $\mathcal W^4$, it is natural to define a space-time manifold in classical Lagrangian mechanics as an $(n+1)$--dimensional fibred manifold over real numbers
\begin{equation}\label{fibracija}
\tau\colon \mathcal Q \longrightarrow \R,
\end{equation}
where the fibers are diffeomorphic to an $n$--dimensional configuration space $Q$ (see \cite{Jo2}).
The points $\mathbf q$ in $\mathcal Q$ are called \emph{events} and the fibers $\tau^{-1}(a)$, $a\in\R$, are called spaces of \emph{simultaneous events}.
We say that the event $\mathbf q_0$  occurred before the event $\mathbf q_1$ if $\tau(\mathbf q_0)<\tau(\mathbf q_1)$.
A \emph{time line} (or \emph{world line})
is a smooth curve $s(t)$, a section
of the fibration \eqref{fibracija} (see Fig. \ref{vremenska}),
\[
\tau(s(t))=t.
\]
A time line $s(t)$, $t\in[a,b]$ is between (or connect) the events $\mathbf q_0$ and $\mathbf q_1$ if $s(a)=\mathbf q_0$
and $s(b)=\mathbf q_1$.

The space of \emph{virtual displacements} is a subbundle of $T\mathcal Q$, the vertical distribution of the fibration \eqref{fibracija}, defined by
\[
\mathcal V=\cup_{\mathbf q\in\mathcal Q}\mathcal V_{\mathbf q}, \qquad \mathcal V_{\mathbf q}=\ker d\tau\vert_{\mathbf q}=
\{\underline\xi\in T_{\mathbf q} \mathcal Q\,\vert\, d\tau\vert_{\mathbf q}(\underline \xi)=0\}.
\]

Since for time lines we have
$
d\tau(\dot{s}(t))=1,
$
we also consider the affine subbundle of $T\mathcal Q$ (the first jet bundle \cite{MP}, see Fig. \ref{vremenska})
\[
\mathcal J=\cup_{\mathbf q\in\mathcal Q}\mathcal J_{\mathbf q}, \qquad
\mathcal J_{\mathbf q}=\{\underline\xi\in T_{\mathbf q} \mathcal Q\,\vert\, d\tau\vert_{\mathbf q}(\underline\xi)=1\}.
\]
It is clear that $\mathcal J$ is diffeomorphic to $\mathcal V$. The Lagrangian of the system $\mathbf L$ is a smooth function defined on the affine bundle $\mathcal J$:
\[
\mathbf L\colon\mathcal J\longrightarrow\R.
\]

The \emph{(global) reference frame} is a trivialisation (see Fig. \ref{vremenska})
\begin{align*}
&\varphi_\alpha\colon \mathcal Q  \longrightarrow Q_\alpha\times \R, \qquad Q_\alpha \cong Q,\\
& \varphi_\alpha(\mathbf q)=(q_\alpha,t_\alpha),
\end{align*}
such that
\[
\tau(\varphi^{-1}_\alpha(q_\alpha,t_\alpha))=t_\alpha+c_\alpha, \qquad c_\alpha\in \R.
\]
In other words,  in the reference frame $\varphi_\alpha$,  we set the time $t_\alpha$ to zero at the space of simultaneous events $\tau^{-1}(c_\alpha)$.

The vertical space $\mathcal V$ at $\mathbf q$ and $\mathcal J_{\mathbf q}$ in the frame $\varphi_\alpha$ can be naturally identified with $T_q Q_\alpha\times \R$ (see Fig. \ref{vremenska}):
\begin{align}
\label{i1}&\underline \eta\in \mathcal V_{\mathbf q} \longleftrightarrow \eta_\alpha\in T_{q_\alpha} Q_\alpha \times\{t_\alpha\}, \quad (\eta_\alpha,0)=d\varphi_\alpha\vert_{\mathbf q}(\underline\eta), \quad (q_\alpha,t_\alpha)=\varphi_\alpha(\mathbf q).\\
\label{i2}&\underline \xi\in \mathcal J_{\mathbf q} \longleftrightarrow \xi_\alpha\in T_{q_\alpha} Q_\alpha \times\{t_\alpha\}, \quad (\xi_\alpha,1)=d\varphi_\alpha\vert_{\mathbf q}(\underline\xi), \quad (q_\alpha,t_\alpha)=\varphi_\alpha(\mathbf q).
\end{align}

If we have two reference frames $\varphi_\alpha$ and $\varphi_\beta$, the transition function
defined by
\[
\phi_{\alpha\beta}=\varphi_\alpha\circ\varphi_\beta^{-1}\colon Q_\beta\times \R \longrightarrow Q_\alpha\times \R
\]
is of the form
\[
(q_\alpha,t_\alpha)=\phi_{\alpha\beta}(q_\beta,t_\beta)=\big(g_{\alpha\beta}(q_\beta,t_\beta),t_\beta+(c_\beta-c_\alpha)\big).
\]

Let $\gamma_\alpha(t_\alpha)$ be a curve in $Q_\alpha$. To $\gamma_\alpha$ we associate the time curve
\[
s(t)=\varphi_\alpha^{-1}((\gamma_\alpha(t_\alpha),t_\alpha))\vert_{t=t_\alpha+c_\alpha},
\]
and vice versa (see Fig. \ref{vremenska}).
Within the identification \eqref{i2}, the Lagrangian $\mathbf L$ in the reference frame $\varphi_\alpha$ is given by
\begin{align*}
& L_\alpha\colon TQ_\alpha\times \R \longrightarrow \R,\\
& L_\alpha(q_\alpha,\dot q_\alpha,t_\alpha)\vert_{q_\alpha=\gamma_\alpha(t_\alpha)}:=\mathbf L(\mathbf q,\dot{\mathbf q})\vert_{\mathbf q=s(t)=\varphi_\alpha^{-1}((\gamma(t_\alpha),t_\alpha))\vert_{t=t_\alpha+c_\alpha}}.
\end{align*}

\begin{figure}[ht]
{\centering
{\includegraphics[width=11cm]{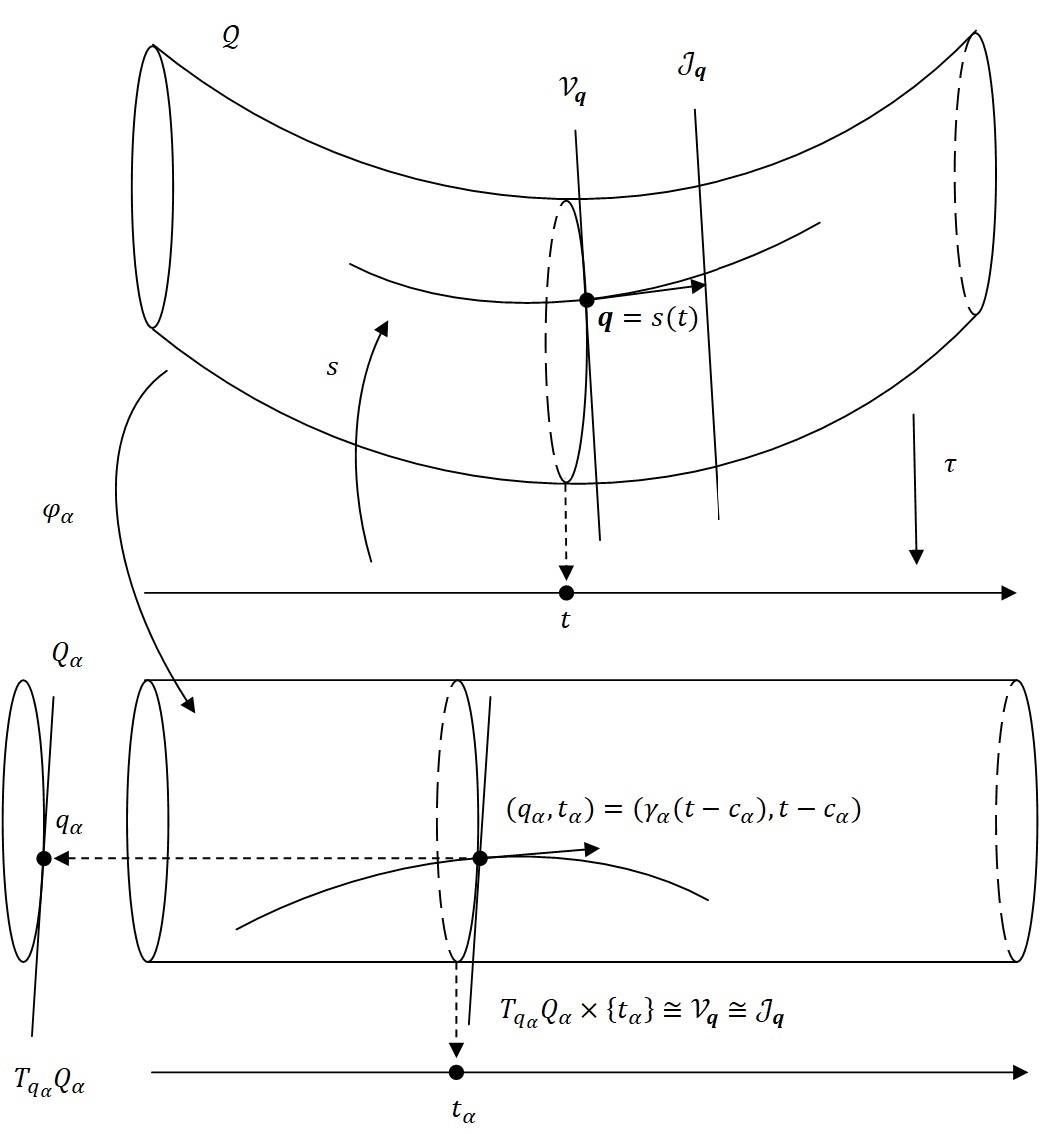}}
\caption{A section (time line) $s$ in the reference frame $\varphi_\alpha$.}\label{vremenska}}
\end{figure}

The \emph{variational derivative} of the Lagrangian $\mathbf L$ in the direction of vector field of virtual displacement $\underline\eta$ along time line $s$ is defined by
\[
\delta\mathbf L(\underline\eta)\vert_{s}:=\delta L_\alpha(\eta_\alpha)\vert_{\gamma_\alpha}.
\]

From Proposition \ref{glavna}, it follows that
the variation derivative does not depend on the reference frame $\varphi_\alpha$ (see \cite{Jo2}).
We thus have an invariant formulation of classical Lagrangian dynamics on the space-time
in the form of the d'Alambert principle:
a time curve $s$ is a motion of the mechanical system $(\mathcal Q,\mathbf L)$ if the
Lagrangian derivative of $\mathbf L$ is equal to zero,
\[
\delta\mathbf L(\underline\eta)\vert_{s}=0,
\]
for all virtual displacements $\underline\eta$ along $s$.

\subsection{Nonholonomic systems}
Let $\underline{\mathcal A}$ be a non-integrable distribution of rank $m+1$ of $T\mathcal Q$ transverse to the distribution of virtual displacements $\mathcal V$.
We define the associated \emph{affine distribution of constraints} $\mathcal A$ of rank $m$, a subspace of the first jet bundle $\mathcal J$, and the distribution of \emph{constrained virtual displacements} $\mathcal D$ of rank $m$, a subspace of the space of virtual displacements $\mathcal V$, respectively by
\[
\mathcal A=\mathcal J\cap \underline{\mathcal A}, \qquad \text{and} \qquad \mathcal D=\mathcal V\cap \underline{\mathcal A}.
\]

A time curve $s(t)$ is \emph{admissible} if $\dot s(t)\in \mathcal A_{s(t)}$ for all $t\in\R$.

\begin{dfn}
An admissible time curve $s$ is a motion of the constrained Lagrangian system $(\mathcal Q,\mathbf L,\mathcal A)$ if
the variational derivative $\delta \mathbf L(\eta)\vert_s$ vanishes for all virtual displacements $\underline{\eta}$,
sections of $\mathcal D$, along $s$:
\[
\delta\mathbf L(\underline\eta)\vert_{s}=0, \qquad \underline{\eta}\in\mathcal D.
\]
\end{dfn}

In the reference frame $\varphi_\alpha$, the nonholonomic system $(\mathcal Q,\mathbf L,\mathcal A)$ corresponds to the
nonholonomic Lagrangian system $(Q_\alpha, L_\alpha, \mathcal A^\alpha)$, where
\[
\mathcal A^\alpha=\{(\mathcal A^\alpha_{t_\alpha},t_\alpha),\,\vert\, t_\alpha\in\R\}\subset TQ_\alpha\times \R
\]
is defined from the natural identification of $\mathcal V$ and $\mathcal J$ in the reference frame $\varphi_\alpha$ (see \eqref{i1},\eqref{i2}). In other words, we have
\[
\mathcal A^\alpha_{t_\alpha}\vert_{q_\alpha} := \big\{\xi\in T_{q_\alpha} Q_\alpha\, \big\vert \,  \xi+\frac{\partial}{\partial t}\in d\varphi_\alpha\big(\mathcal A\vert_{\varphi_\alpha^{-1}(q_\alpha,t_\alpha)}\big)   \big\}.
\]

\section{The Voronec equations and time-dependent fibration}\label{s5}

\subsection{The Voronec equations}
We recall the Voronec equations for nonholonomic systems \cite{Vor1902} (in \cite{Vor1902} Voronec derived the equations
from the Lagrange-d'Alembert principle for the case of time-independent homogeneous constraints, here we adopt the notation of Bilimovic's student Demchenko \cite{Dem1924, DGJ}). Let $q=(q^1,\dots,q^{n})$ be local coordinates of the configuration space $Q$.
Consider a nonholonomic system with the Lagrangian $L=L(q,\dot q,t)$, non-potential forces (or generalized forces) $F_i=F_i(q,\dot q,t)$, which correspond to the coordinates $q^i$, $i=1,\dots,n$, and time-dependent, non-homogeneous, nonholonomic constraints of the form
\begin{equation}\label{voronec1}
\dot q^{m+\nu}=\sum_{i=1}^m a_{\nu i} (q,t) \dot q^i + a_\nu (q,t),  \qquad \nu=1,\dots,r=n-m.
\end{equation}

Let $L_c$ be the Lagrangian $L$ after imposing the constraints \eqref{voronec1}. Let ${J}_{\nu}$ be the partial derivatives of $L$
with respect to $\dot q^{m+\nu}$, $\nu=1,2,\dots,r$, restricted to the constrained subspace. We assume that the constraints \eqref{voronec1} are imposed \emph{after} the differentiation and obtain
\begin{align}
\label{voronecLc} & L_c(q^1,\dots,q^{n},\dot q^1,\dots,\dot q^m,t)=L(q,\dot q,t)\vert_{\dot q^{m+\nu}=\sum_{i=1}^m a_{\nu i} (q,t) \dot q^i + a_\nu (q,t)},  \\
\label{voronecKn}& {J}_\nu(q^1,\dots,q^{n},\dot q^1,\dots,\dot q^m,t)=\frac{\partial L\,\,\,}{\partial \dot q^{m+\nu}}(q,\dot q,t)\vert_{\dot q^{m+\nu}=\sum_{i=1}^m a_{\nu i} (q,t) \dot q^i + a_\nu (q,t)}.
\end{align}

The Voronec equations of motion of the given noholonomic system can be represented in the following closed form:
\begin{equation}\label{voronec3}
 \frac{d}{dt}\frac{\partial L_c}{\partial \dot q^i}=\frac{\partial L_c}{\partial q^i}+F_i
 +\sum_{\nu=1}^r a_{\nu i}\big( \frac{\partial L_c}{\partial q^{m+\nu}}+F_{m+\nu}\big)+ \sum_{\nu=1}^r {J}_\nu\big(\sum_{j=1}^m A_{ij}^{(\nu)} \dot q^j+A_{i}^{(\nu)}\big),
\end{equation}
$i=1,\dots,m$.
The components $A_{ij}^{(\nu)}$ and $A_i^{(\nu)}$ are functions of the time $t$ and the coordinates $q^1,\dots,q^{n}$ given by
\begin{align}
\label{voronec4}& A_{ij}^{(\nu)}=  \big(\frac{\partial a_{\nu i}}{\partial q^j}+\sum_{\mu=1}^r a_{\mu j}\frac{\partial a_{\nu i}}{\partial q^{m+\mu}}\big)
                   -\big(\frac{\partial a_{\nu j}}{\partial q^i}+\sum_{\mu=1}^r a_{\mu i}\frac{\partial a_{\nu j}}{\partial q^{m+\mu}}\big), \\
\label{voronec5} &  A_{i}^{(\nu)}= \big(\frac{\partial a_{\nu i}}{\partial t}+\sum_{\mu=1}^r a_{\mu }\frac{\partial a_{\nu i}}{\partial q^{m+\mu}}\big)
                   -\big(\frac{\partial a_{\nu}}{\partial q^i}+\sum_{\mu=1}^r a_{\mu i}\frac{\partial a_{\nu}}{\partial q^{m+\mu}}\big), \quad i,j=1,\dots,m.
\end{align}

\begin{rem}\label{voronecPrinciple}
The equations can be rewritten more compactly using a formal expression known as the \emph{Voronec principle} (see \cite{Vor1902, Dem1924, DGJ}).
As Demchenko noted, the Voronec principle is similar to Hamilton's variational principle of least action, although the system is not variational (see pages 16--19 of \cite{Dem1924}). In the original notation: A trajectory $q(t)=(q^1(t),\dots,q^{n}(t))$, $t\in [t_1,t_2]$, is a motion of the nonholonomic system with Lagrangian $L(q,\dot q,t)$,
generalized forces $F_i=F_i(q,\dot q,t)$, and constraints \eqref{voronec1}, if
\begin{equation}\label{vor6}
\int_{t_1}^{t_2}\Big[\sum_{i=1}^n \big(\frac{\partial L}{\partial q^i}-\frac{d}{dt}\frac{\partial L}{\partial \dot q^i}+F_i\big)\delta q^i+\sum_{\nu=1}^k {J}_\nu\big(\frac{d}{dt}\delta q^{m+\nu}-\delta\dot q^{m+\nu}\big)\Big]dt=0,
\end{equation}
for all virtual displacements $\delta q^1,\dots,\delta q^{m}$ that are equal to zero for $t=t_1$ and $t=t_2$, but otherwise arbitrary. The remaining variations $\delta q^{m+1},\dots,\delta q^{n}$ are determined from the homogeneous constraints
\begin{equation}\label{vor7}
\delta q^{m+\nu}=\sum_{i=1}^m a_{\nu i} \delta q^i, \qquad \nu=1,2,\dots,r.
\end{equation}
Here $\frac{d}{dt}\delta q^{n+\nu}-\delta\dot q^{n+\nu}$ are calculated according to the relations
\eqref{voronec1}, \eqref{vor7}, using the rule
\[
\frac{d}{dt}\delta q^{i}-\delta\dot q^{i}=0, \qquad i=1,2,\dots,m.
\]
\end{rem}

\subsection{The Ehresmann connections for time-dependent nonholomic systems}

Let us consider a time-dependent non-holonimic Lagrangian system $(Q,L,\mathcal A)$.
Following Bloch, Krishnaprasad, Marsden and Murray \cite{BKMM},
we assume that the configuration space $Q$ has the structure of a time-dependent fiber bundle $\pi_{t}\colon Q\to S$ over a base manifold $S$ and that the distribution of virtual displacements $\mathcal D_t$ is transversal to the fibers of $\pi_{t}$:
\begin{equation}\label{transverzalnost1}
T_q Q=\mathcal D_t\vert_q \oplus \mathcal W_t\vert_q, \qquad \mathcal W_t\vert_q=\ker d\pi_{t}(q).
\end{equation}
In other words: $\mathcal D_t$ is a time-dependent \emph{horizontal space}, while
$\mathcal W_t$ \emph{vertical space} at $q$ the Ehresmann connection of the fibration (for the current fibration, $\mathcal W_t$ is not time-dependent, but this is relaxed in the Section \ref{pokretniVoronec}).

Next, we consider the distribution $\underline{\mathcal A}$ of the extended configuration space $Q\times\R$ as the horizontal space of the Ehresmann connection with respect to the fibration
\begin{equation*}\label{Eh1}
\underline{\pi}\colon Q\times \R\to S\times \R, \qquad \underline{\pi}(q,t)=(\pi_{t}(q),t).
\end{equation*}
Namely, the transversality \eqref{transverzalnost1} is equivalent to the transversality of $\underline{\mathcal A}$
\begin{equation}\label{transverzalnost2}
T_{(q,t)}(Q\times\R)=\underline{\mathcal A}_{(q,t)} \oplus \mathcal W_t\vert_q.
\end{equation}

The distribution $\underline{\mathcal A}$ can be regarded as the kernel of a vertical-valued one-form $\underline{\Theta}$ on $Q\times\R$ that satisfies

\begin{enumerate}[label=(\roman*)]

\item $\underline{\Theta}_{(q,t)}\colon T_{(q,t)} (Q\times \R) \to\mathcal W_q$ is a linear mapping, $(q,t)\in Q\times\R$;

\item $\underline{\Theta}$ is a projection: $\underline{\Theta}_{(q,t)}({\underline{\xi}}_{(q,t)})={\underline{\xi}}_{(q,t)}$, for all ${\underline{\xi}}_{(q,t)}\in\mathcal W_q$.

\end{enumerate}

By $\underline{\xi}^v=\underline{\Theta}(\underline{\xi})$ and $\underline{\xi}^h=\underline{\xi}-\underline{\xi}^v$ we denote the vertical and horizontal components of a vector field on the extended configuration space $\underline{\xi}\in\mathfrak{X}(Q\times \R)$.

The \emph{curvature} $\underline{B}$ of the connection \eqref{transverzalnost2} is a vertical-valued two-form defined by
\begin{equation}\label{podvucenoB}
\underline{B}(\underline{\xi},\underline{\eta})=-\underline{\Theta}([\underline{\xi}^h,\underline{\eta}^h]), \qquad \underline{\xi}, \underline{\eta}\in\mathfrak{X}(Q\times \R).
\end{equation}

Let, as above, $\dim Q=n=m+r$ and $\dim S=m$. There exist local ``adapted'' coordinates $q=(q^1,\dots,q^{n})$ on $Q$ such that
$(q^1,\dots,q^m)$ are the local coordinates on $S$,
the projections
$\pi_{t}\colon Q\to S$ and $\underline{\pi}\colon Q\times \R\to S\times \R$ have the form
\begin{align*}
& \pi_{t}\colon (q^1,\dots,q^m,q^{m+1},\dots,q^{n}) \longmapsto (q^1,\dots,q^m),\\
& \underline{\pi}\colon (q^1,\dots,q^m,q^{m+1},\dots,q^{n},t) \longmapsto (q^1,\dots,q^m,t)
\end{align*}
and the constraints defining $\mathcal A_t$ are given by \eqref{voronec1}.
Then we have locally
\begin{align*}
& \underline{\Theta}=\sum_{\nu=1}^r {\underline{\Theta}^\nu} \frac{\partial}{\partial q^{m+\nu}},
\qquad {\underline{\Theta}^\nu}=dq^{n+\nu}-\sum_{i=1}^m a_{\nu i} dq^i-a_\nu dt, \\
& {\underline{\xi}}^h=\big(\sum_{l=1}^{n}{\xi}^l\frac{\partial}{\partial q^{l}}+{\xi^{n+1}}\frac{\partial}{\partial t}\big)^h\\
&\quad=\sum_{i=1}^{m}{\xi}^i\frac{\partial}{\partial q^{i}} + {\xi^{n+1}}\frac{\partial}{\partial t}
+\sum_{\nu=1}^r\big(\sum_{i=1}^{m}a_{\nu i}{\xi}^i+a_\nu {\xi^{n+1}}\big)\frac{\partial}{\partial q^{m+\nu}},\\
& {\underline{\xi}}^v=\big(\sum_{l=1}^{n}{\xi}^l\frac{\partial}{\partial q^{l}}{+{\xi^{n+1}}\frac{\partial}{\partial t}}\big)^v=
\sum_{\nu=1}^r\big({\xi}^{m+\nu}-\sum_{i=1}^{m}a_{\nu i}{\xi}^i-a_\nu {\xi^{n+1}}\big)\frac{\partial}{\partial q^{m+\nu}},\\
& \underline{B}(\frac{\partial}{\partial q^{i}},\frac{\partial}{\partial q^{j}})=\sum_{\nu=1}^r B^\nu_{ij}\frac{\partial}{\partial q^{m+\nu}},\\
& \underline{B}(\frac{\partial}{\partial q^{i}},\frac{\partial}{\partial t})=\sum_{\nu=1}^r B^\nu_{i,t}\frac{\partial}{\partial q^{m+\nu}}, \quad
\underline{B}(\frac{\partial}{\partial t},\frac{\partial}{\partial q^i})=\sum_{\nu=1}^r B^\nu_{t,i}\frac{\partial}{\partial q^{m+\nu}},
\end{align*}
where
\begin{equation}\label{KRIVINE}
B_{ij}^\nu(q,t)=A^{(\nu)}_{ij}(q,t), \quad B_{i,t}^\nu(q,t)=A^{(\nu)}_{i}(q,t),  \quad B_{t,i}^\nu(q,t)=-A^{(\nu)}_{i}(q,t),
\end{equation}
 and $A^{(\nu)}_{ij}(q,t)$, $A^{(\nu)}_i(q,t)$ are derived from the Voronec equations \eqref{voronec4} \eqref{voronec5}, $i,j=1,\dots,m$, while the other components of the curvature are zero. To the author's knowledge, the above curvature interpretation of the Voronec coefficients is \eqref{voronec4} \eqref{voronec5} for time-dependent non-homogeneous constraints was first given by Bak\v sa \cite{Baksa2012}.

\begin{exm}\label{primer3}
In Example \ref{primer1} we have $m=1$, $n=2$, $(q^1,q^2)=(y,x)$ and
\[
B_{1,t}=\frac{\partial a}{\partial t}-\frac{\partial b}{\partial y}+b\frac{\partial a}{\partial x}-a\frac{\partial b}{\partial x}, \qquad B_{t,1}=-B_{1,t}.
\]
The distribution $\underline{\mathcal A}$ is contact if and only if the curvature $\underline{B}$ is different from zero.
\end{exm}

\begin{rem}
Similarly, $\mathcal D_t$ can be seen as the kernel of a
time-dependent vector-valued one form $\Theta_t$ on $Q$ with properties (i) and (ii) for the fibration $\pi$.
Locally, $\Theta_t$ is given by the same equation as $\underline{\Theta}$, where $\underline{\Theta}^\nu$ is replaced by
$\Theta^\nu_t=dq^{n+\nu}-\sum_{i=1}^m a_{\nu i} dq^i$ (see e.g. \cite{BKMM} for the time-independent case).
Let
$B_t$ be the corresponding  time-dependent curvature
\[
{B_t}({\xi},{\eta})=-{\Theta_t}([{\xi}^h,{\eta}^h]), \qquad \xi,\eta\in \mathfrak{X}(Q).
\]

Since $\mathcal D_t$ is tangent to the foliation
$\{(q,t)\in Q\times\R\,\vert\,t=const\}$ and $\mathcal D_t\vert_q\subset \underline{\mathcal A}_{(q,t)}$, from the definition of the curvature we have
\[
B_t\vert_{\mathcal D_t\vert_q\times\mathcal D_t\vert_q}=\underline{B}\vert_{\mathcal D_t\vert_q\times\mathcal D_t\vert_q},
\quad B_t(\mathcal W_t\vert_q,T_q Q)\equiv 0, \quad \underline{B}(\mathcal W_t\vert_q,T_{(q,t)} Q\times\R)\equiv 0.
\]

The components of the curvature $B_t$ of the time-dependent connection $D_t$ are the same as the corresponding components of the curvature $\underline{B}$
\[
B_t(\frac{\partial}{\partial q^{i}},\frac{\partial}{\partial q^{j}})=\sum_{\nu=1}^r B^\nu_{ij}\frac{\partial}{\partial q^{m+\nu}}, \qquad i,j=1,\dots,n.
\]
\end{rem}

\begin{rem}
Note that, even if the constraint are homogeneous ($\mathcal A_t=\mathcal D_t$),
the components of the curvature $B_{t,i}$ are different from zero in the time-dependent case.
\end{rem}

Let $\theta_t$ be a vertical-valued vector field given by
\[
\theta_t=-\big(\frac{\partial}{\partial t}\big)^v=-\underline{\Theta}\big(\frac{\partial}{\partial t}\big).
\]

\begin{lem}
The affine subspace $\mathcal A_t\vert_q\subset T_q Q$ is the image of the tangent space $T_qQ$ under the composition of affine and linear mappings
\begin{equation}\label{afinaProjekcija}
\mathrm{pr}_t\colon \eta\in T_q Q \longmapsto
\big(\eta+\frac{\partial}{\partial t}\big)^h-\frac{\partial}{\partial t}=\eta^h+\theta_t \in \mathcal A_t\vert_q.
\end{equation}
\end{lem}

Note that $\mathrm{pr}_t$ is the identity mapping  restricted to ${\mathcal A_t\vert_q}$. In other words, $\mathrm{pr}_t$ is an affine projection.
In particular, we have
\[
\mathcal A_t=\mathcal D_t+{\mathbf\theta}_t.
\]

Let $L_c\colon TQ\times \R\to\R$ be the \emph{constrained Lagrangian} defined by
\[
L_c(q,\dot q,t):=L(q,\mathrm{pr}_t(\dot q),t)=L(q,\dot q^h+\theta_t,t).
\]

Locally,   the vertical vector field $\theta_t$ and the affine projection $\mathrm{pr}_t$ are
\begin{align*}
&\theta_t=-\big(\frac{\partial}{\partial t} \big)^v=\sum_{\nu=1}^r a_\nu \frac{\partial}{\partial q^{m+\nu}},\\
&\mathrm{pr}_t(\eta)=\big(\sum_{l=1}^{n}\eta^l\frac{\partial}{\partial q^{l}}\big)^h+\theta_t=\sum_{i=1}^{m}\eta^i\frac{\partial}{\partial q^{i}}
+\sum_{\nu=1}^r\big(\sum_{i=1}^{m}a_{\nu i}\eta^i+a_\nu\big)\frac{\partial}{\partial q^{m+\nu}},
\end{align*}
the constrained Lagrangian $L_c$ is given by \eqref{voronecLc}, and we have:\footnote{ Note that we first apply the fiber derivative of $L$ and then the affine projection $\mathrm{pr}_t$. If we first take the affine projection $\mathrm{pr}_t$  and then the fiber derivative, the result would be $\mathbb FL_c$.}
\begin{align*}
\mathbb FL(q,\mathrm{pr}_t(\dot q),t)&\big(\underline{B}\big(\dot q+\frac{\partial}{\partial t},\eta\big)\big)\vert_{\gamma(t)}
=\mathbb FL(q,\mathrm{pr}_t(\dot q),t)\big(B_t(\dot q,\eta)+\underline{B}\big(\frac{\partial}{\partial t},\eta  \big)\big)\vert_{\gamma(t)}\\
=&\sum_{\nu=1}^m \frac{\partial L}{\partial \dot q_{m+\nu}}\vert_{\dot q^{m+\nu}=\sum_{i=1}^m a_{\nu i} (q,t) \dot q^i + a_\nu (q,t)}\big(B^\nu_t(\dot q,\eta)+\sum_{s=1}^n B_{t,s}^\nu\eta^s \big)\vert_{\gamma(t)}.
\end{align*}

Therefore, the Voronec principle \eqref{vor6} can be expressed in invariant form as follows.

\begin{thm}\label{invVoronec}
An admissible curve $\gamma$ is a motion of the constrained Lagrangian system $(Q,L,\mathbf F,\mathcal A)$ if and only if
\begin{equation}\label{voronec6}
\delta L_c(\eta_{t})\vert_{\gamma{(t)}}+\mathbf F(\eta_{t})\vert_{\gamma{(t)}}=\mathbb FL(q,\mathrm{pr}_t(\dot q),t)(\underline{B}(\dot q+\frac{\partial}{\partial t},\eta_{t}))\vert_{\gamma{(t)}},
\end{equation}
for all virtual displacements
$
\eta_{t}=\sum_{s=1}^{n} \eta^s \frac{\partial}{\partial q_s}\in \mathcal D_t\vert_{\gamma(t)}
$
along $\gamma (t)$.  Here $\underline{B}$ is the curvature \eqref{podvucenoB} and $\mathrm{pr}_t$ is the affine projection \eqref{afinaProjekcija}.
 \end{thm}

Note that the restrictions of the constrained Lagrangian $L_c$ and the Lagrangian $L$ on the affine distribution coincide:
\[
L_c\vert_{\mathcal A_t}=L\vert_{\mathcal A_t}.
\]

\subsection{Natural mechanical systems}

Let us consider the Lagrangian
\[
L(q,\dot q,t)=\frac12 K_t(\dot q,\dot q)+\Delta_t(\dot q)-V(q,t)
\]
of a natural mechanical system.
Then
the constrained Lagrangian reads
\begin{align*}
L_c(q,\dot q,t)=\frac12 K_t(\dot q^h,\dot q^h)+\big(K_t(\theta_t,\dot q^h)+\Delta_t(\dot q^h)\big)+\big(\frac12 K_t(\theta_t,\theta_t)+\Delta_t(\theta_t)-V(q,t)\big).
\end{align*}

Since the fiber derivative is given by \eqref{FB}, the right-hand side of the system \eqref{voronec6} is
\begin{align*}\label{RHS}
\mathbb FL(q,&\mathrm{pr}_t(\dot q),t)(\underline{B}(\dot q+\frac{\partial}{\partial t},\eta))=K_t(\dot q^h, B_t(\dot q,\eta))+K_t(\theta_t, B_t(\dot q,\eta))\\
\nonumber &+\Delta_t(B_t(\dot q,\eta))
+K_t(\dot q^h,\underline{B}\big(\frac{\partial}{\partial t},\eta\big))+K_t(\theta_t,\underline{B}\big(\frac{\partial}{\partial t},\eta\big))+\Delta_t(\underline{B}\big(\frac{\partial}{\partial t},\eta\big)).
\end{align*}

\subsection{The Ehresmann connection in the moving frame}\label{pokretniVoronec}

Next, we consider the moving frame $g_t\colon M\to Q$.
As a result, the moving configuration space $M$ has the structure of a time-dependent fiber bundle $\Pi_t\colon M\to S$  defined by
\[
\Pi_t(x):=\pi_{t}(g_t(x)),
\]
The distribution of the virtual displacements $\mathbf D_t$ (see equation \eqref{pokretneDistribucije}) is transverse to the fibers of $\Pi_t$:
\begin{equation*}\label{transverzalnost1m}
T_x M=\mathbf D_t\vert_x \oplus \mathbf W_t\vert_x, \qquad \mathbf W_t\vert_x=\ker d\Pi_t(x)=dg_t^{-1}(\mathcal W_t\vert_q),
\end{equation*}
and $\mathbf D_t$ and $\mathbf W_t$ are \emph{horizontal space} and \emph{vertical space} at $x$ of the Ehresmann connection of the time-dependent fibration $\Pi_t$.

As above, we consider the distribution $\underline{\mathbf A}$ (see equation \eqref{PodvucenoMasnoA}) of the extended configuration space $M\times\R$ as the horizontal space of the Ehresmann connection
\begin{equation}\label{transverzalnost2m}
T_{(x,t)}(M\times\R)=\underline{\mathbf A}_{(x,t)} \oplus \mathbf W_t\vert_x.
\end{equation}
related to the fibration
\[
\underline{\Pi}\colon M\times \R\to N\times \R, \qquad \underline{\Pi}(x,t)=(\Pi_t(x),t).
\]

Again, we define  a vertical-valued one-form $\underline{\mathbf T}$ on $M\times\R$, analogous to $\underline{\Theta}$, a vertical-valued vector field $\mathbf t_t$,
\[
\mathbf t_t=-\big(\frac{\partial}{\partial t}\big)^v=-\underline{\mathbf T}\big(\frac{\partial}{\partial t}\big),
\]
the curvature $\underline{\mathbf B}$ of the connection \eqref{transverzalnost2m},
\begin{equation}\label{podvucenoMasnoB}
{\underline{\mathbf B}(\underline{\xi},\underline{\eta})=-\underline{\mathbf T}([\underline{\xi}^h,\underline{\eta}^h]),  \qquad \underline{\xi}, \underline{\eta}\in\mathfrak{X}(M\times \R),}
\end{equation}
the affine projection
\begin{equation}\label{MasnaProjekcija}
\mathbf{pr}_t\colon \xi\in T_x M \longmapsto
\big(\xi+\frac{\partial}{\partial t}\big)^h-\frac{\partial}{\partial t}=\xi^h+\mathbf t_t\in \mathbf A_t\vert_x,
\end{equation}
and the constrained Lagrangian $l_c$:
\[
l_c(x,\dot x,t):=l(x,\mathbf{pr}_t (\dot x),t).
\]

\begin{lem}
The following diagram is commutative
\begin{equation*}
\xymatrix@R28pt@C28pt
{
T_x M \ar@{->}[d]^{\mathbf{pr}_t}  \ar@{->}[r] &  T_q Q  \ar@{->}[d]^{\mathrm{pr}_t} \\
\mathbf A_t\vert_x   \ar@{->}[r] &  \mathcal A_t\vert_q
}
\end{equation*}
where the horizontal lines denote the mapping $\xi\mapsto dg_t(\xi)+\omega_t\vert_q$, $q=dg_t(x)$ and its restriction to $\mathbf A_t\vert_x\subset T_x M$.
\end{lem}

Whence, equivalently, the constrained Lagrangian $l_c$ can be defined from the constrained Lagrangian $L_c$ in the fixed frame:
\begin{equation*}\label{noviLagranzijanC}
l_c(x,\dot x,t):=L_c(q,\dot q,t)\vert_{q=g_t(x),\dot q=dg_t(\dot x)+\omega_t(g_t(x))}.
\end{equation*}

As a result we have

\begin{prop}\label{movingVoronec}
An admissible curve $\Gamma$ is a motion of the constrained Lagrangian system in the moving frame $(M,l,\mathbf f,\mathbf A)$ if
\begin{equation*}\label{voronec7}
\delta l_c(\xi_{t})\vert_{\Gamma(t)}+\mathbf f(\xi_{t})\vert_{\Gamma{ (t)}}=
\mathbb Fl(x,{\mathbf{pr}_t(\dot x)},t)(\underline{\mathbf B}(\dot x+\frac{\partial}{\partial t},\xi_{t}))\vert_{\Gamma{(t)}},
\end{equation*}
for all virtual displacements
$
\xi_{t}=\sum_{s=1}^{n} \xi^s \frac{\partial}{\partial x_s}\in \mathbf D_t\vert_{\Gamma(t)}
$
along $\Gamma$. Here $\underline{\mathbf B}$ is the curvature  \eqref{podvucenoMasnoB} and $\mathbf{pr}_t$ is the affine projection \eqref{MasnaProjekcija}.
\end{prop}

\section{Time-dependent Chaplygin systems}\label{s6}

\subsection{Definitions}
In addition to the assumptions from Section \ref{s5}, we now assume that the time-dependent fibration $\pi_t: Q\to S$ is determined by a time-dependent free action of an $r$--dimensional Lie group $G$ on $Q$.  Furthermore, for every fixed time $t\in\R$, both  the affine constraint space $\mathcal A_t$ and the distribution of virtual displacements $\mathcal D_t$ are  assumed to be $G$--invariant.
Then the $G$--orbit $\mathcal O_t(q)$ through a point $q$ is the fiber $\pi_t^{-1}(\pi_t(q))$, $t\in\R$ (see Fig. \ref{lifts}).
The decomposition
\begin{equation*}\label{transverzalnost1Ch}
T_q Q=\mathcal D_t\vert_q \oplus \mathcal W_t\vert_q, \qquad \mathcal W_t\vert_q=\ker d\pi_t(q)=T_q(\mathcal O_t(q)),
\end{equation*}
is $G$--invariant and $\mathcal D_t$ is a principal connection of a time-dependent principal bundle $G\longrightarrow Q\longrightarrow S=Q/G$ (here a time $t$ is fixed).
On the other hand, the decomposition
\begin{equation}\label{transverzalnost2Ch}
T_{(q,t)}(Q\times\R)=\underline{\mathcal A}_{(q,t)} \oplus \mathcal W_t\vert_q,
\end{equation}
is also $G$--invariant and $\underline{\mathcal A}$ is a principal connection of the principal bundle
\begin{equation*}
\xymatrix@R28pt@C28pt{
G \ar@{^{}->}[r]  & Q\times \R \ar@{^{}->}[d]^{\underline{\pi}}  \\
  & S\times \R }, \qquad \underline{\pi}(q,t)=(\pi_t(q),t).
\end{equation*}

Let $y=\pi_t(q)$. For vectors $\xi\in T_y S$, $\underline{\xi}\in T_{(y,t)}(S\times \R)$, the \emph{horizontal lifts}
$\mathbb{H}^q_t({\xi})\in {\mathcal D}_t\vert_q$ and $\mathbb{H}^{(q,t)}(\underline{\xi})\in \underline{\mathcal A}_{(q,t)}$ are the unique
vectors, such that
\[
d{\pi_t}(\mathbb{H}^q_t(\xi))={\xi}, \qquad d\underline{\pi}(\mathbb{H}^{(q,t)}(\underline{\xi}))=\underline{\xi}.
\]

In addition, we define the \emph{affine horizontal lift} of $\xi\in T_y S$ by (see Fig. \ref{lifts})
\[
\mathbb{A}^q_t(\xi)=\mathrm{pr}_t \mathbb H^q_t\big(\xi\big)=\mathbb H^q_t\big(\xi\big)+\theta_t\vert_q\in\mathcal A_t\vert_q.
\]

\begin{figure}[ht]
{\centering
{\includegraphics[width=9cm]{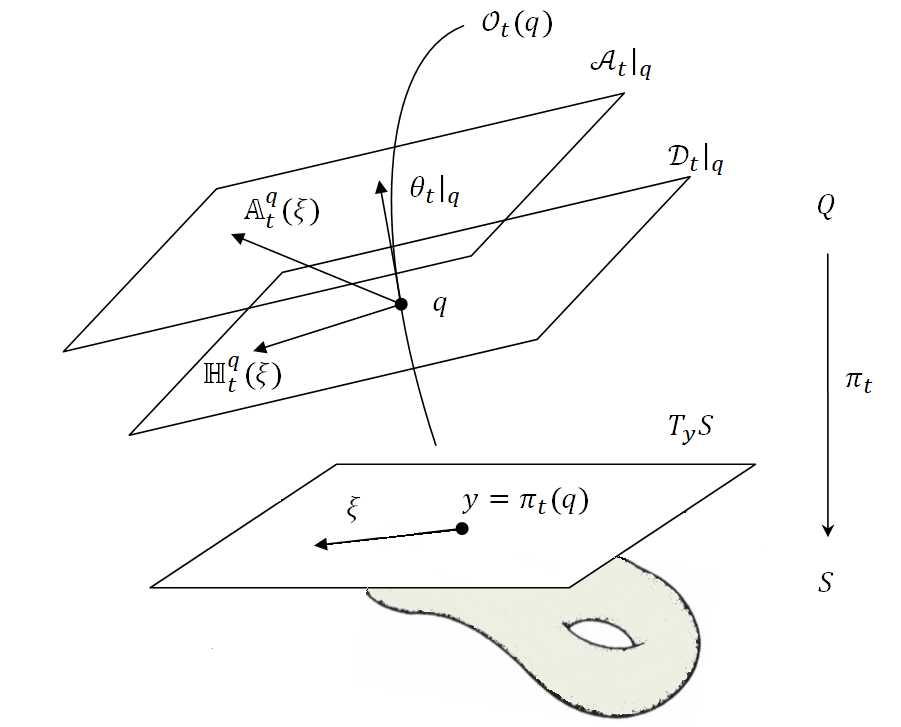}}
\caption{The horizontal and affine horizontal lift of a vector $\xi\in T_y S$.}\label{lifts}}
\end{figure}

Note that the Ehresmann curvature $\underline{B}$ of the principal connection \eqref{transverzalnost2Ch} via the identification $\mathcal W_t\cong \mathfrak g$ can be understood as a $\mathfrak g$-valued 2-form on $Q$.

By analogy with time-independent $G$--Chaplygin systems \cite{Koi1992, BKMM} and $G$--Chaply\-gin systems with gyroscopic forces \cite{DGJ2} we define.

\begin{dfn}\label{DEF}
Let $L$ be a time-dependent $G$--invariant Lagrangian and let $\mathbf F$ be a field of time-dependent $G$--invariant non-potential forces on $Q$.
The \emph{reduced Lagrangian} $L_{red}(y,\dot y,t)$, the \emph{reduced field of non-potential forces} $\mathbf F_{red}(y,\dot y,t)$,
and the \emph{JK term} on $S$ are respectively defined by:
\[
L_{red}(y,\dot y,t):=L(q,\mathbb A^q_t(\dot y),t),  \qquad \mathbf F_{red}(y,\dot y,t)(\xi):=\mathbf F(\mathbb A^q_t(\dot y),t)(\mathbb H_t^q(\xi)),
\]
and
\[
{\mathbf{JK}}(y,\dot y,t)(\xi):=
\mathbb FL(q,\mathbb A_t^q(\dot y),t)(\underline{B}\big(\mathbb H^{(q,t)}\big(\dot y+\frac{\partial}{\partial t}\big),\mathbb H^q_t(\xi))\big),
\qquad \xi\in T_y S,
\]
where $q$ is any element of the fiber $\pi_t^{-1}(y)$.
\end{dfn}

\begin{dfn}
We refer to $(Q,L,\mathbf F,\mathcal A,G)$ as a \emph{time-dependent $G$--Chaplygin system}, and to
$(S,L_{red},\mathbf F_{red}, {\mathbf{JK}})$ as a \emph{reduced time-dependent $G$--Chaplygin system}.
\end{dfn}

We summarize the above consideration in the following statement.

\begin{thm}\label{redukcijaSistema}
A solution $\gamma(t)$ of the time-dependent $G$--Chaplygin system
$(Q,L$, $\mathbf F,\mathcal A,G)$ projects to  a solution
\[
{\gamma_{red}}(t)=\pi_t(\gamma(t))
\]
 of the reduced time-dependent $G$--Chaplygin system
$(S,L_{red},\mathbf F_{red}, {\mathbf{JK}})$. Conversely, let ${\gamma_{red}}(t)$ be a solution of the reduced system
\begin{equation}
\delta L_{red}(\xi_t)\vert_{\gamma_{red}(t)}+\mathbf F_{red}(\xi_t)\vert_{\gamma_{red}(t)}= \mathbf{JK}(y,\dot y,t)(\xi_t)\vert_{\gamma_{red}(t)},
\label{ChaplyginRed}
\end{equation}
for all virtual displacements $\xi_t\in T_{\gamma_{red}(t)}S$ along $\gamma_{red}(t)$,
with the initial conditions ${\gamma_{red}}(0)=y_0$, $\dot\gamma_{red}(0)=Y_0\in T_{y_0} S$
and let $q_0\in \pi^{-1}(y_0)$. Then the affine horizontal lift $\gamma(t)$ of ${\gamma_{red}}(t)$ through $q_0$ is the solution of the original system  with the initial conditions  $\gamma(0)=q_0$, $\dot\gamma(0)=\mathbb A_t^{q_0}(Y_0)\in \mathcal A_t\vert_{q_0}$.
\end{thm}
 \begin{proof}
The equations \eqref{voronec6} are $G$--invariant and they reduce to $TS$. By Definition \ref{DEF}, the terms
$\delta L_c(\eta_t)\vert_{\gamma(t)}$, $\mathbf F(\eta_t)\vert_\gamma(t)$, and $\mathbb FL(q,\mathrm{pr}_t(\dot q),t)(\underline{B}(\dot q+\frac{\partial}{\partial t},\eta_t))\vert_{\gamma(t)}$ correspond to $\delta L_{red}(\xi_t)\vert_{\gamma_{red}(t)}$, $\mathbf F_{red}(\xi_t)\vert_{\gamma_{red}(t)}$, and $\mathbf{JK}(y,\dot y,t)(\xi_t)\vert_{\gamma_{red}(t)}$, respectively, where $\xi_t\vert_{\gamma_{red}(t)}=d\pi(\eta_t\vert_{\gamma(t)})$. This proves the first part of the statement.
For the second part, for a given virtual displacement $\xi_t$ along a  solution $\gamma_{red}(t)$ of the reduced system \eqref{ChaplyginRed}, we take a horizontal lift $\eta_t\vert_{\gamma(t)}=\mathbb{H}^{\gamma(t)}_t(\xi_t\vert_{\gamma_{red}(t)})$ as a virtual displacement along the horizontal affine lift $\gamma(t)$.
  \end{proof}

In local coordinates $(y^1,\dots,y^m)$ on $S=Q/G$, the equations \eqref{ChaplyginRed} read
\begin{equation*}
\frac{d}{dt}\frac{\partial l}{\partial
\dot y_i}=\frac{\partial L}{\partial y_i}+F_{red,i}(y,\dot y,t)-\mathbf{JK}_i(y,\dot y,t), \qquad i=1,\dots,m,
\label{ChaplyginRedLoc}
\end{equation*}
 where
\[
\mathbf{JK}=\sum_{i=1}^m {\mathbf{JK}}_{i}(y,\dot y,t) dy^i, \qquad  \mathbf F_{red}=\sum_{i=1}^m F_{red,i}(y,\dot y,t) dy^i.
\]

\begin{rem}\label{AbelovSlucaj}
In the case the group $G$ is Abelian, we do not need an invariant formulation of the Voronec equations, which we obtain in Theorem \ref{invVoronec} and Proposition \ref{movingVoronec}. One can perform the reduction procedure directly by using the classical Voronec equations \eqref{voronec3}, where the Lagrangian $L$ and the constraints \eqref{voronec1} do not depend on $q^{m+\nu}$ and the action of the group $G$ is simply given by the translations in the coordinates $q^{m+\nu}$, $\nu=1,\dots,n-m$
(classical Chaplygin systems \cite{Chaplygin1912}).
\end{rem}

\subsection{Natural mechanical systems}\label{NMS}

The above construction can be related to the standard construction for time-independent natural mechanical systems.
In the case that we have a natural mechanical system with the Lagrangian
\[
L(q,\dot q,t)=\frac12 K_t(\dot q,\dot q)+\Delta_t(\dot q)-V(q,t),
\]
we define

\begin{dfn}
Let $K_t$, $\Delta_t$, and $V(q,t)$ be a $G$--invariant metric, 1-form, and a potential  on $Q$. The \emph{reduced metric} $K_{red,t}$, the \emph{reduced 1-form} $\Delta_{red,t}$ and the \emph{reduced potential} $V_{red}$ on $S$ are defined by:
\begin{align*}
&K_{red,t}(\xi_1,\xi_2)\vert_y=K_t(\mathbb H^q_t(\xi_1),\mathbb H^q_t(\xi_2))\vert_q,  \\
&\Delta_{red,t}(\xi)=K_t(\theta_t\vert_q,\mathbb H^q_t(\xi))+\Delta_{t}(\mathbb H^q_t(\xi))\big),\\
&V_{red}(y,t)=V(q,t)-\frac12 K_t(\theta_t\vert_q,\theta_t\vert_q)-\Delta_t(\theta_t\vert_q), \qquad y=\pi(q).
\end{align*}
\end{dfn}

The reduced Lagrangian is then
\begin{align*}
L_{red}(y,\dot y,t)=\frac12 K_{red,t}(\dot y,\dot y))+\Delta_{red,t}(\dot y)-V_{red}(y,t).
\end{align*}

On the other hand, the JK term has the form
\begin{align*}\label{JKRHS}
{\mathbf{JK}}(y,\dot y,t)(\xi)={\mathbf{JK}}^{(2)}(y,\dot y,t)(\xi)+{\mathbf{JK}}^{(1)}(y,\dot y,t)(\xi)+{\mathbf{JK}}^{(0)}(y,\dot y,t)(\xi),
\end{align*}
where
\begin{align*}
{\mathbf{JK}}^{(2)}(y,\dot y,t)(\xi)=&K_t\big(\mathbb H^q_t(\dot y), B_t(\mathbb H^q_t(\dot y),\mathbb H^q_t(\xi))\big)\\
{\mathbf{JK}}^{(1)}(y,\dot y,t)(\xi)=&K_t\big(\theta_t\vert_q, B_t(\mathbb H^q_t(\dot y),\mathbb H^q_t(\xi)\big)\\
&\nonumber +\Delta_t\big(B_t(\mathbb H^q_t(\dot y),\mathbb H^q_t(\xi)\big)+K_t\big(\mathbb H^q_t(\dot y),\underline{B}\big(\frac{\partial}{\partial t},\mathbb H^q_t(\xi)\big)\big)\\
{\mathbf{JK}}^{(0)}(y,\dot y,t)(\xi)=&K_t\big(\theta_t,\underline{B}\big(\frac{\partial}{\partial t},\mathbb H^q_t(\xi)\big)\big)
+\Delta_t\big(\underline{B}\big(\frac{\partial}{\partial t},\mathbb H^q_t(\xi)\big)\big).
\end{align*}

Here the first term is quadratic in velocities and is a time-dependent version of the well-studied $(0,3)$-tensor in nonholonomic mechanics
\begin{equation*}\label{eq:Sigma}
\Sigma_t(\xi_1,\xi_2,\xi_3)\vert_y=K_t\big(\mathbb H^q_t(\xi_1),B_t(\mathbb H^q_t(\xi_2),\mathbb H^q_t(\xi_3))\big)\vert_q, \qquad q\in\pi^{-1}(y).
\end{equation*}
The $(0,3)$-tensor $\Sigma_t$ is skew-symmetric with respect to the second and third argument, and
\[
{\mathbf{JK}}^{(2)}(y,\dot y,t)(\xi)=\Sigma_t(\dot y,\dot y,\xi).
\]

%Apart from $\Sigma_t$, in the time-independent case it is convenient to introduce (1,2)-tensor fields $\mathbf B$ and $\mathbf C$ defined by (see \cite{Koi1992,CCL2002})
%\[
%\Sigma(\xi_1,\xi_2,\xi_3)=K_t(\mathbf B(\xi_1,\xi_2),\xi_3)=K_t(\xi_1,\mathbf C(\xi_2,\xi_3)).
%\]
%Note that $\mathbf C$ is equal to the negative gyroscopic tensor $\mathcal T$ defined in \cite{Naranjo2019b, Naranjo2020}.

After \cite{Chaplygin1912}, one of the central problems in the study of natural-mechanical Chaplygin systems is the Chaplygin-Hamiltonization using the time reparamitrazition (see e.g \cite{BBM2013, CCL2002, GajJov2019b, Naranjo2019b, Naranjo2020, DGJ2} and references therein).
In \cite{BBM2013} and \cite{DGJ2} the problem is extended to the case of an addition of a gyroscopic term $\mathbf F_{red}$.

In the time-dependent case, the simplest situation is when we use a distinguished reference frame in which $\theta_t\equiv 0$ and the linear term in the Lagrangian also vanishes $\Delta_t \equiv 0$. Then the JK term of zero degree ${\mathbf{JK}}^{(0)}$ vanishes and the linear JK term simplifies to
\[
{\mathbf{JK}}^{(1)}(y,\dot y,t)(\xi)=K_t(\mathbb H^q_t(\dot y),\underline{B}\big(\frac{\partial}{\partial t},\mathbb H^q_t(\xi)\big)).
\]

Probably, the study of the problem of Hamiltonization of the time-dependent Chaplygin system should be started with these assumptions.
The following is a first step in this direction.

\subsection{Rolling of a disc over a circle of a given variable radius - a nonholonomic system with one-dimensional time-dependent constraint distribution}
An illustrative example is the rolling without sliding of a disc with radius $r$, mass $m$ and centre $S$ over a vertical circle with  variable radius $R(t)$ as a given function of time and centre at point $O$. The configuration space is $\R^2\{s_1,s_2\}\times S^1\{\psi\}$. The coordinates $(s_1,s_2)$ determine the position of the centre of the disc
$S$, $\overrightarrow{OS}=s_1 \vec e_1+ s_2\vec e_2$, and the angle $\psi$ is the angle of rotation $B(\psi)$ between the fixed reference frame $[O,\vec e_1,\vec e_2]$ and the reference frame $[S,\vec E_1,\vec E_2]$ attached to the disc (see Fig. \ref{disk}):
\[
\vec E_1=\cos\psi \vec e_1+\sin\psi\vec e_2, \qquad \vec E_2=-\sin\psi \vec e_1+\cos\psi\vec e_2.
\]

In other words, the coordinates of a point $X$ in two reference frames are related by the Euclidean motion
\begin{equation}\label{AfTr}
\begin{pmatrix}
q_1 \\
q_2
\end{pmatrix}=
\begin{pmatrix}
\cos\psi & -\sin\psi\\
\sin\psi & \cos\psi
\end{pmatrix}
\begin{pmatrix}
x_1 \\
x_2
\end{pmatrix}+
\begin{pmatrix}
s_1\\
s_2
\end{pmatrix},
\end{equation}
$\overrightarrow{OX}=q_1 \vec e_1+ q_2\vec e_2$, $\overrightarrow{SX}=x_1 \vec E_1+ x_2\vec E_2$.

\begin{figure}[ht]
{\centering
{\includegraphics[width=7.5cm]{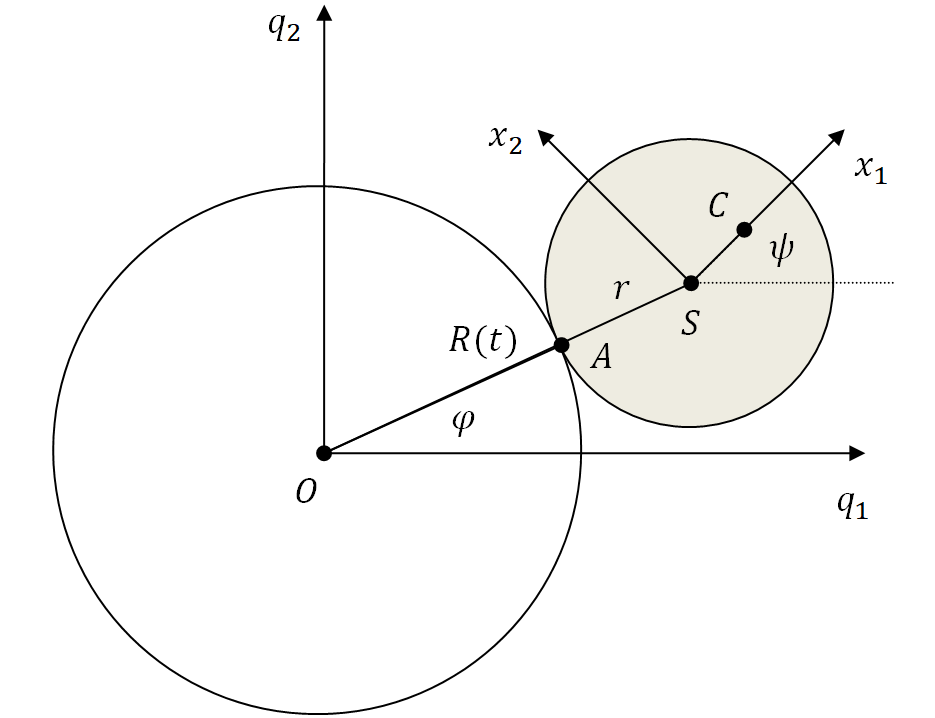}}
\caption{Rolling of a disc of radius $r$, center $S$ and mass center $S$ over a vertical circle of variable radius $R(t)$ and center $O$.}\label{disk}}
\end{figure}

Let us consider a point $X(x_1,x_2)$ that is fixed in the frame $[S,\vec E_1,\vec E_2]$. From \eqref{AfTr} it follows that its velocity in the frame $[O,\vec e_1,\vec e_2]$
is
\begin{align*}
\begin{pmatrix}
v_1 \\
v_2
\end{pmatrix}=& \frac{d}{dt}\big(
\begin{pmatrix}
\cos\psi & -\sin\psi\\
\sin\psi & \cos\psi
\end{pmatrix}
\begin{pmatrix}
x_1 \\
x_2
\end{pmatrix}+
\begin{pmatrix}
s_1\\
s_2
\end{pmatrix}\big)\\
=&
\dot\psi
\begin{pmatrix}
-\sin\psi & -\cos\psi\\
\cos\psi & -\sin\psi
\end{pmatrix}
\begin{pmatrix}
x_1 \\
x_2
\end{pmatrix}
+\begin{pmatrix}
\dot s_1\\
\dot s_2
\end{pmatrix}
\\
=&
\dot\psi
\begin{pmatrix}
-\sin\psi & -\cos\psi\\
\cos\psi & -\sin\psi
\end{pmatrix}
\begin{pmatrix}
\cos\psi & -\sin\psi\\
\sin\psi & \cos\psi
\end{pmatrix}^{-1}
\begin{pmatrix}
X_1\\
X_2
\end{pmatrix}
+
\begin{pmatrix}
\dot s_1\\
\dot s_2
\end{pmatrix}
\\
=&
\begin{pmatrix}
0 & -\dot\psi\\
\dot\psi &  0
\end{pmatrix}
\begin{pmatrix}
X_1\\
X_2
\end{pmatrix}
+
\begin{pmatrix}
\dot s_1\\
\dot s_2
\end{pmatrix}, \qquad \text{where} \qquad \overrightarrow{SX}=X_1 \vec e_1+ X_2\vec e_2,
\end{align*}
that is,
\begin{equation}\label{brzina}
(v_1,v_2)=\dot\psi(-X_2,X_1)+(\dot s_1,\dot s_2).
\end{equation}

We have a holonomic constraint
\[
f(s,t)=s_1^2+s_2^2-(R(t)+r)^2=0,
\]
and instead of the coordinates $(s_1,s_2)$ we can use the angular coordinate $\varphi$, so that
\[
(s_1,s_2)=((R+r)\cos\varphi,(R+r)\sin\varphi).
\]
For the configuration space $Q$ we can therefore take a torus $Q=S^1\times S^1\{\varphi,\psi\}$.

Let $A$ be a point of contact. The non-sliding condition at  $A$
means that the velocity of the point $A$, considered as a point of the circle with coordinates $(R\cos\varphi,R\sin\varphi)$ in the frame $[O,\vec e_1,\vec e_2]$,
\begin{equation}\label{contact}
(v_{A1},v_{A2})=\dot R\cdot (\cos\varphi,\sin\varphi),%+R\dot\varphi(-\sin\varphi,\cos\varphi),
\end{equation}
coincides to the velocity of the point $A$, which is considered as the point of the disc with the coordinates $(x_{A1},x_{A2})$ in the frame $[S,\vec E_1,\vec E_2]$.
From \eqref{brzina} and \eqref{contact} we get
\begin{align*}
\dot R(\cos\varphi,\sin\varphi)=   %+R\dot\varphi(-\sin\varphi,\cos\varphi)=\\
\dot \psi(r\sin\varphi,-r\cos\varphi)+\dot R(\cos\varphi,\sin\varphi)+(R+r)\dot\varphi(-\sin\varphi,\cos\varphi),
\end{align*}

Therefore, we obtain a homogeneous constraint
\begin{equation}\label{VEZA}
\mathcal A_t=\mathcal D_t\colon \qquad r\dot\psi-(R+r)\dot\varphi=0.
\end{equation}

If the constraint \eqref{VEZA} is stationary, $\dot R \equiv 0$, then it is integrable. For a given initial condition, the system is a 1-dimensional time-independent holonomic system and therefore completely integrable. Otherwise, we have a nonholonomic system (see Example \ref{primer1}) with a constraint defining the Ehresmann connections for the fibrations $\pi\colon Q\to S^1$, $\pi(\varphi,\psi)=\varphi$ and $\pi(\varphi,\psi)=\psi$. We will use the first variant in the following.

Without loss of generality, we assume that the centre of mass $C$ of the disc has the coordinates $(x_{C1},x_{C2})=(c,0)$. Then
\begin{align*}
&(X_{C1},X_{C2})=(c\cos\psi,c\sin\psi), \\
& (q_{C1},q_{C2})=((R+r)\cos\varphi,(R+r)\sin\varphi)+(c\cos\psi,c\sin\psi),
\end{align*}
and the gravitational potential energy of the system is
\[
v_g(\varphi,\psi,t)=mg q_{C2}=mg\big((R(t)+r)\sin\varphi+c\sin\psi\big),
\]
while the kinetic energy is given by
\begin{align*}
T(\varphi,\psi,\dot\varphi,\dot\psi,t)=& \frac12 m \big(\dot s_1^2+\dot s_2^2\big)+\frac12 I\dot\psi^2+m\dot\psi(X_{C1}\dot s_2-X_{C2}\dot s_1)\\
=& \frac12 m \big((R+r)^2\dot\varphi^2+\dot R^2\big)+\frac12 I\dot\psi^2 + c m \dot\psi\big(\dot R(\cos\psi\sin\varphi\\
& -\sin\psi\cos\varphi)+\dot\varphi(R+r)(\cos\psi\cos\varphi +\sin\psi\sin\varphi)\big),
\end{align*}
where $I$ is the moment of inertia of the disc with respect to the point $S$.

As a result, we obtain a time-dependent natural mechanical system on a torus with the Lagrangian
$L=\frac12 K_{t}+\Delta_{t}-V$, where
\begin{align*}
&K_{t}\big((\dot\varphi,\dot\psi),(\dot\varphi,\dot\psi)\big)=m(R+r)^2\dot\varphi^2+I\dot\psi^2+2c m \dot\psi\dot\varphi(R+r)\cos(\varphi-\psi),\\
&\Delta_{t}\big((\dot\varphi,\dot\psi)\big)= c m \dot\psi\dot R\sin(\varphi-\psi),\\
&V(\varphi,\psi,t)=mg\big((R(t)+r)\sin\varphi+c\sin\psi\big)-\frac12 m\dot R^2.
\end{align*}
Note that the term $\frac12 m\dot R^2$ can be removed from the potential as it does not depend on the coordinates.

The system simplifies considerably if the centre of mass $C$ coincides with the centre of the disc $S$, $c=0$, i.e. the discs is balanced. Then the Lagrangian
\begin{align*}
L(\varphi,\psi,\dot\varphi,\dot\psi,t)=\frac12 \big(m(R(t)+r)^2\dot\varphi^2+I\dot\psi^2\big)-
mg(R(t)+r)\sin\varphi
\end{align*}
does not depend on $\psi$ and we obtain a time-dependent $S^1$-Chaplygin system. Since the dimension of the reduced space is one-dimensional, it is convenient to use the Voronec equations directly (see Remark \ref{AbelovSlucaj}).

We define the function:
\begin{equation}\label{a(t)}
 a(t):=\frac{R(t)+r}{r}.
\end{equation}

 Then the constraint, the constrained (reduced) Lagrangian $L_c$ (see \eqref{voronecLc}), the term ${J}_\psi$ (see \eqref{voronecKn}) and the curvature   $B_{\varphi,t}^\psi=A_\varphi^{(\psi)}$  (see \eqref{voronec5} and \eqref{KRIVINE}) are  given by:
\begin{align*}
&\dot\psi=a(t)\dot\varphi,\\
&L_c(\varphi,\dot\varphi,t)=\frac12 (mr^2+I)a^2(t)\dot\varphi^2-
mr g a(t)\sin\varphi,\\
& {J}_\psi (\varphi,\dot\varphi,t)=\frac{\partial L}{\partial \dot\psi}\vert_{\dot\psi=a(t)\dot\varphi}=Ia(t)\dot\varphi,\\
& B_{\varphi,t}^\psi=A_\varphi^{(\psi)}=\dot a(t).
\end{align*}

The JK term therefore only has the linear velocity term
\[
\mathbf{JK}=\mathbf{JK}^{(1)}=-{J}_\psi A_{\varphi}^{(\psi)}{d\varphi}= -Ia(t)\dot a(t)\dot\varphi d\varphi
\]
and the (reduced) Voronec equation
\begin{equation*}
 \frac{d}{dt}\frac{\partial L_c}{\partial \dot\varphi}=\frac{\partial L_c}{\partial\varphi} + {J}_\psi A_{\varphi}^{(\psi)}
\end{equation*}
has the form
\begin{equation}\label{1reduced}
 (mr^2+I)a^2(t)\ddot\varphi + 2(mr^2+I)a(t)\dot a(t) \dot\varphi  =-mr g a(t)\cos\varphi + Ia(t)\dot a(t)\dot\varphi.
\end{equation}

Note that if $I$ is equal to zero,  the equation \eqref{1reduced} describes the mathematical pendulum with variable length. For $a(t+T)=a(t)$ this is a basic example of a parametric resonance (see Ch. 5 in \cite{Ar} and \cite{BSL})).

Let us recall the Chaplygin Hamiltonisation of $G$-Chaplygin autonomous systems on the base space $S$ with local coordinates $y=(y^1,\dots,y^m)$: we are looking for a time reparametrization $d\tau= \mathcal N(y)dt= \mathcal N(y^1,\dots,y^m)dt$ so that in the new time $\tau$ the Euler-Lagrange equations with $\mathbf{JK}^{(2)}$ term become the usual Euler-Lagrange equations with respect to the new reduced Lagrangian obtained from the original one by using the transformation $dy/dt=\mathcal N dy/d\tau$.
(see e.g. \cite{Chaplygin1912, DGJ2}).
\footnote{We use the same symbol $\tau$ as for the time-mapping \eqref{fibracija} in Section \ref{s4}, as it is convenient for the Chaplygin reparametrization.}

To perform Hamiltonization for time-dependent $G$-Chaplygin systems, we need to modify the Chaplygin multiplier method to include a time-reparametrization
where $\mathcal N$ depends also on $t$.

\begin{lem}\label{hamiltonizacija}
Let us consider a one-dimensional Lagrangian problem with the Lagrangian
\[
L(y,\dot y,t)=\frac12 K(t)\dot y^2-V(y,t)
\]
 and a non-potential force $\mathbf F=B(t)\dot y dy$, where $y\in \R$, or $y$ is an angular coordinate on a circle $S^1$. For each $t_0\in \R$ there is a neighborhood $t_0\in (t_1,t_2)$ and a time-reparametrization $t=z(\tau)$, $\tau\in (\tau_1,\tau_2)$, so that in the new time $\tau$ the system becomes a Lagrangian problem without non-potential forces with the Lagrangian
\[
L^*(y,y',\tau):=L(y,\dot y, t)\vert_{\dot y=\mathcal N(\tau)y', t=z(\tau)}=\frac12 K(z(\tau))(\mathcal N(\tau) y')^2-V^*(y,\tau),
\]
 $V^*(y,\tau)=V(y,z(\tau))$, where $'$ denotes the derivative with respect to the new time $\tau$ and
 $\mathcal N(\tau)=1/z'$.
 \end{lem}

 \begin{proof}
The Euler--Lagrange equation of the system is
 \begin{equation}\label{EL1}
 \frac{d}{dt}\big(K\dot y\big)=K\ddot y+\dot K\dot y=-\frac{\partial V}{\partial y}+B(t)\dot y.
 \end{equation}

Let us introduce a time reparametrization $t=z(\tau)$. We denote the inverse of the mapping $z$ by $u$: $\tau=u(t)$. Then
\[
\N=\frac{1}{z'}=\dot u \circ z.
\]

 Let $g(t)$ be an arbitrary function in time $t$ and $G(\tau)=g(z(\tau))$ be the corresponding function in the new time variable. Then $\dot g=dg/dt$ and $G'=dG/d\tau$ are related by
\begin{equation}\label{izvodi}
 \dot g\circ z=\frac{dG}{d\tau}\frac{d\tau}{dt}\circ z=G'\cdot \dot u\circ z=\mathcal N G' \quad \text{and}\quad
 G'\circ u=\frac{dg}{dt}\frac{dt}{d\tau}\circ u=\dot g\cdot  z'\circ u.
\end{equation}
Thus, for a motion $g_1(t)=y(t)$ and its derivative $g_2(t)=\dot y(t)$, from the first identity in \eqref{izvodi}, we get
 \[
 \dot y\circ z=\mathcal N Y', \qquad \ddot y\circ z=\mathcal N(\mathcal NY')'= \mathcal N^2 Y''+\mathcal N\mathcal N'Y'.
 \]

From now on, as in the formulation of the statement, we simply use $y(\tau)$ instead of $Y(\tau)=y(z(\tau))$.
Then, the equation \eqref{EL1} in time $\tau$ takes the form
 \begin{equation}\label{EL2}
 K(z(\tau))\mathcal N^2 y''+K(z(\tau))\mathcal N\mathcal N'y' +{\big(\mathcal N \frac{d}{d\tau}\big(K(z(\tau))\big)\big)\mathcal N} y'=-\frac{\partial V^*}{\partial y}+B(z(\tau))\mathcal N y'.
 \end{equation}

On the other hand, the Euler--Lagrange equation of the Lagrangian $L^*$ with respect to the new time $\tau$,
\[
\frac{d}{d\tau}\frac{\partial L^*}{\partial y'}=\frac{\partial L^*}{\partial y},
\]
is
 \begin{equation*}\label{EL3}
 K(z(\tau))\mathcal N^2 y''+ \big(2K(z(\tau))\mathcal N\mathcal N'+\N^2\frac{d}{d\tau}\big(K(z(\tau))\big)\big)y'=-\frac{\partial V^*}{\partial y}.
 \end{equation*}

The last equation is equivalent to the equation \eqref{EL2} if and only if
\[
B(z(\tau))\mathcal N y' =-K(z(\tau))\mathcal N\mathcal N' y'.
\]

Since $\mathcal N\ne 0$ and the identity applies to all $y'$, we obtain the differential equation
\begin{equation}\label{1dim}
\frac{d\mathcal N}{d\tau}= f(z(\tau)), \qquad     f(t):=-\frac{B(t)}{K(t)}.
\end{equation}

If we consider $\tau$ as a function of $t$,
\[
\N(\tau(t))=(\N\circ u)(t)=(\dot u \circ z \circ u)(t)=\dot u(t),
\]
 and using the second identity in \eqref{izvodi} for the function $G(\tau)=\mathcal N(\tau)$, we obtain:
\begin{align*}
 \N'\circ u=\ddot u \cdot \big(z' \circ u\big)=
\frac{\ddot u}{\dot u}=\frac{d\ln(\dot u(t))}{dt}.
\end{align*}

Therefore, from \eqref{1dim} we get the equation
\begin{equation}\label{2dim}
\frac{d\ln(\dot u(t))}{dt}=(f\circ z\circ u) (t)=f(t),
\end{equation}
which can be easily solved
\[
\dot u(t)=\N_0\exp(\int_{t_0}^t f(s)ds), \qquad \dot u(t_0)=\N_0>0.
\]

Finally, we find a time-reparametrization $z=z(\tau)$ by inverting the integral
\[
\tau-\tau_0=u(t)=\int_{t_0}^{t} \dot u(s)ds, \qquad u(t_0)=\tau_0.
\]
\end{proof}

\begin{exm}
Let us consider the harmonic oscillator $L(y,\dot y)=\frac12 \dot y^2-\frac12 \omega^2 y^2$ with the damping force $\mathbf F=-B_0\dot y dy$, $B_0=const> 0$.
Let us apply the construction described in Lemma \ref{hamiltonizacija}. The equation \eqref{2dim} is as follows
\[
\frac{d\ln(\N(t))}{dt}=B_0.
\]

The solution with the initial conditions $\dot u\vert_{t=0}=\N_0> 0$ is $\dot u(t)=\N_0e^{B_0t}$. From this follows,
\begin{equation}\label{zamVr}
\begin{aligned}
&\tau-\tau_0=\int_{0}^{t} \N_0 e^{B_0s} ds=\frac{\N_0}{B_0} e^{B_0t}-\frac{\N_0}{B_0}, \\
& t=z(\tau)=\frac{1}{B_0}\ln\frac{B_0(\tau-\tau_0)+\N_0}{\N_0},
\end{aligned}
\end{equation}
where $\tau\in(\tau_0-\N_0/B_0,\infty)$ and $t\in(-\infty,\infty)$. From $\dot u(t)=\N_0e^{B_0t}$ and \eqref{zamVr} we obtain
\[
\N(\tau)=\dot u\circ z(\tau)=B_0(\tau-\tau_0)+\N_0.
\]

This means that the harmonic oscillator with the damping force $\mathbf F=-B_0\dot y dy$ becomes the standard Lagrangian system with the Lagrangian
\[
L^*(y,y',\tau)=\frac12\big(B_0(\tau-\tau_0)+\N_0\big)^{2}y'^2-\frac12 \omega^2 y^2,
\]
for the new time $\tau\in(\tau_0-\N_0/B_0,\infty)$.
\end{exm}

Let us return to the reduced problem of rolling a disc over a circle with a variable radius.

\begin{thm}\label{hamiltonizacija*}
 Let $R(t)$ be a smooth real function, $R(t)+r>0$, and let $a(t)$ defined by \eqref{a(t)} be bounded from above.
Consider the rolling without sliding of a balanced disc of radius $r$ and mass $m$
over a vertical circle of a given  variable radius $R(t)$.
Then, a time-reparametrization $t=z(\tau)$, where $z\colon \R\to \R$ is obtained by inverting the integral
\begin{equation}\label{inverz}
\tau= u(t)=\int_{0}^t {a(s)}^{-\alpha} ds, \qquad \alpha=\frac{I}{mr^2+I},
\end{equation}
transforms the reduced system \eqref{1reduced} into the Lagrangian problem with the Lagrangian
 \[
L_c^*(\varphi,\varphi',\tau)=\frac12 (mr^2+I)a^2(z(\tau))\big(\mathcal N(\tau)\varphi'\big)^2-
mr g a(z(\tau))\sin\varphi.
 \]
In particular, if $R(t)$ is  periodic, the transformed system is also  periodic.
\end{thm}

\begin{proof}
The equation \eqref{2dim} takes the form
\begin{equation*}
\label{1dimRED}
\frac{d\ln(\dot u(t))}{dt}=-\frac{Ia(t)\dot a(t)}{ (mr^2+I)a^2(t)}=-\frac{I}{ (mr^2+I)}\frac{d}{dt}\ln(a(t)).
\end{equation*}

Thus,
\[
\dot u(t)=\N_0\big(\frac{a(t_0)}{a(t)}\big)^\alpha.
\]

We can set $t_0=0$, $\dot u(0)=\N_0=a(0)^{-\alpha}$, $u(0)=\tau_0=0$.

From the definition of $\alpha$ and since $a(t)$ is bounded, we have $0<\alpha<1$ and $0<a(t)<a_{max}$, for some positive constant $a_{max}$. A time reparametrization $t=z(\tau)$ is
given by inverting the integral \eqref{inverz}.  Whence, we get
\begin{align*}
&\tau  > a_{max}^{-\alpha} t, \qquad t>0,\\
&\tau  < a_{max}^{-\alpha} t, \qquad t <0,
\end{align*}
and $\tau=u(t)$ is the invertible mapping defined for all $t,\tau\in\R$.

If $a(t)$ is periodic with a period  $t_0$, it is clear that
for
\begin{equation*}\label{period}
\tau_0=u(t_0)=\int_{0}^{t_0} {a(s)}^{-\alpha} ds,
\end{equation*}
we will have $u(t+t_0)=u(t)+\tau_0$, $z(\tau+\tau_0)=z(\tau)+t_0$, and, therefore, $\mathcal N(\tau+\tau_0)=\mathcal N(\tau)$, $a(z(\tau+\tau_0))=a(z(\tau))$.
Whence,
$L^*_c(\varphi,\varphi',\tau+\tau_0)=L^*_c(\varphi,\varphi',\tau)$.
\end{proof}

\begin{rem}
Alternatively, by introducing $p={\partial L^*_c}/{\partial \varphi'}$ and the Hamiltonian function
\[
H^*(\varphi,p,\tau)=\frac{1}{2(mr^2+I)\big(a(z(\tau))\mathcal N(\tau)\big)^2} p^2+mr g a(z(\tau))\sin\varphi,
\]
we can write the equations in Hamiltonian form on $T^* S^1\{\varphi,p\}$,
\begin{align*}
&\varphi'=\frac{\,\,\partial H^*}{\partial p}=\frac{1}{(mr^2+I)\big(a(z(\tau))\mathcal N(\tau)\big)^2}p, \\
& p'=-\frac{\,\,\partial H^*}{\partial\varphi}=-mr g a(z(\tau))\cos\varphi.
\end{align*}
\end{rem}

\begin{rem}{\color{black}
The above example indicates that, in general, for an arbitrary dimension of the reduced configuration space $S=Q/G$,
an almost complete analogy between classical Chaplygin Hamiltonization and the almost-Hamiltonian description provided in Theorem 5.1 \cite{DGJ2} may not hold, since
the time-rescaling $\tau=z(t,y^1,\dots,y^n)$ is quite restrictive.  In the study of Hamiltonization and integrability, it seems natural to pass from the category of symplectic structures on the cotangent bundle $T^{*}S$ to the category of cosymplectic or contact structures on \(T^*S \times \mathbb{R}\) (see, e.g., \cite{JL,J2025} and references therein) and the rescaling of the associated Reeb vector fields.}
\end{rem}

\subsection*{Acknowledgments}
The author would like to thank the reviewer for a number of valuable comments and suggestions which helped to  improve significantly  the exposition of the results. The author is grateful to Vladimir Dragovi\'c and Borislav Gaji\'c for useful discussions. The research is based on the long-standing joint experience in teaching classical mechanics and symplectic geometry at the Mathematical Institute SANU. The author is also grateful to Tijana \v Sukilovi\'c for several corrections.
The research was supported by the Project no. 7744592, MEGIC "Integrability
and Extremal Problems in Mechanics, Geometry and Combinatorics" and IntegraRS of the Science Fund
of Serbia.

\end{document}